\documentclass[british]{amsart}
\usepackage{textcomp}
\usepackage[applemac]{inputenc}
\usepackage{array}
\usepackage{units}
\usepackage{mathrsfs}
\usepackage{multirow}
\usepackage{amstext}
\usepackage{amsthm}
\usepackage{amssymb}
\usepackage{geometry}
\usepackage{wasysym}

\makeatletter

\newcommand{\lyxmathsym}[1]{\ifmmode\begingroup\def\b@ld{bold}
  \text{\ifx\math@version\b@ld\bfseries\fi#1}\endgroup\else#1\fi}

\providecommand{\tabularnewline}{\\}

\numberwithin{equation}{section}
\numberwithin{figure}{section}
\theoremstyle{plain}
\newtheorem{thm}{\protect\theoremname}
\theoremstyle{definition}
\newtheorem{example}[thm]{\protect\examplename}
\theoremstyle{definition}
\newtheorem{defn}[thm]{\protect\definitionname}
\theoremstyle{plain}
\newtheorem{prop}[thm]{\protect\propositionname}
\theoremstyle{plain}
\newtheorem{lem}[thm]{\protect\lemmaname}
\theoremstyle{remark}
\newtheorem{rem}[thm]{\protect\remarkname}
\theoremstyle{plain}
\newtheorem{cor}[thm]{\protect\corollaryname}

\usepackage{babel}

\makeatother

\usepackage{babel}
\providecommand{\corollaryname}{Corollary}
\providecommand{\definitionname}{Definition}
\providecommand{\examplename}{Example}
\providecommand{\lemmaname}{Lemma}
\providecommand{\propositionname}{Proposition}
\providecommand{\remarkname}{Remark}
\providecommand{\theoremname}{Theorem}

\begin{document}
\title{\textsf{gonosomic algebras : an extension of gonosomal algebras}}
\maketitle
\begin{center}
\textsc{R. Varro} \smallskip{}
\par\end{center}

\textbf{Abstract}.{\small{} Gonosomal algebras provide an algebraic
model for all sex-determination systems and a wide variety of
sex-related genetic phenomena, but they do not allow for the
modelling of partial or total genetic sterility. In this paper
we introduce gonosomic algebras which extend the notion of gonosomal
algebras and allow us to take into account genetic sterility.
Conditions under which gonosomic algebras are not gonosomal and
several algebraic constructions of these algebras are given and
illustrated with genetic examples. To each gonosomic algebra,
is associated a quadratic operator, noted $W$, that gives the
state of the offspring population at the birth stage. Next from
$W$ we define the operator $V$ which gives the relative frequency
distribution within a population. We show how the equilibrium
points and the various stability notions of equilibrium points
are preserved from $W$ to $V$.}{\small\par}

\smallskip{}

\smallskip{}

{\small\textbf{Key words.}}{\small{} Gonosomal algebra, Gonosomic
algebra, duplication of algebra, Bisexual population, Genetic
sterility, Quadratic operator, Gonosomic operator, equilibrium
point, limit point.}{\small\par}

\section{Introduction}

In the animal and plant kingdoms, reproduction is a fundamental
process for the persistence and evolution of species. There is
an incredible diversity of strategies used for the reproduction.
We can distinguish two forms of reproduction: sexual and asexual.
Prevalence of sexual reproduction in eukaryotic multicellular
organisms (metazoa) is estimated at more than 99\%. Sexual reproduction
is characterised by an exchange of genetic material between two
organisms. In this exchange, the organism that gives is said
to be male and the one that receives is said to be female. This
results in a partition of the population into two classes called
sexes: males and females. In most bisexual species sex determination
systems are based on sex chromosomes also called gonosomes. According
to the sex chromosomes observed in a species we have five sex
determination systems: $XY$, $WZ$, $X0$, $Z0$ or $WXY$,
to which multiple variants must be added. These systems and the
genotypes that determine sex are summarised in the table below.\medskip{}

\begin{center}
\begin{tabular}{c|c|c|c|c|c|c|}
\multicolumn{2}{c|}{sex- determination systems} & $XY$ & $WZ$ & $X0$ & $Z0$ & $WXY$\tabularnewline
\hline 
\hline 
\multirow{2}{*}{Genotypes} & $\mbox{Female}$ & $XX$ & $WZ$ & $XX$ & $Z0$ & $XX,WX,WY$\tabularnewline
\cline{2-7}
 & $\mbox{Male}$ & $XY$ & $ZZ$ & $X0$ & $ZZ$ & $XY,YY$\tabularnewline
\end{tabular}
\par\end{center}

\medskip{}

The study of all these sex determination systems can be unified
by using the non-associative structures introduced by I.M.H.
Etherington \cite{Ether-40}. The first algebraic model was proposed
by Etherington \cite{Ether-41} for a gonosomal diallelic gene
in the $XY$-system, it was extended to diallelic case with mutation
in \cite{Gonsh-60}, to multiallelic case in \cite{Gonsh-65,WB-74,WB-75},
to a single multiallelic locus, completely or partially linked
to a sex determining locus \cite{Holg-70}. The second model
is due to Gonshor \cite{Gonsh-73} by introducing the concept
of sex-linked duplication. In \cite{LR-10} the authors introduced
a more general definition: the evolution algebras of a bisexual
population (\emph{EABP}). But several genetic situations are
not representable by \emph{EABP} which led to the introduction
of gonosomal algebras \cite{RV-15} which generalise \textit{EABP}. 

\label{def:Gonosomal} Given a commutative field $\mathbb{K}$
with characteristic not 2, a $\mathbb{K}$-algebra $A$ is \textit{gonosomal}
of type $\left(n,m\right)$ if it admits a basis $\left(e_{i}\right)_{1\text{\ensuremath{\le}}i\text{\ensuremath{\le}}n}\cup\left(\widetilde{e}_{j}\right)_{1\text{\ensuremath{\le}}j\text{\ensuremath{\le}}m}$
such that for all $1\leq i,j\leq n$ and $1\leq p,q\leq m$ we
have 
\begin{alignat*}{2}
i) &  & \quad e_{i}e_{j}=\widetilde{e}_{p}\widetilde{e}_{q} & =0,\\
ii) &  & e_{i}\widetilde{e}_{p}\;=\;\widetilde{e}_{p}e_{i} & =\sum_{k=1}^{n}\gamma_{ipk}e_{k}+\sum_{r=1}^{m}\widetilde{\gamma}_{ipr}\widetilde{e}_{r},\\
iii) &  & \quad\sum_{k=1}^{n}\gamma_{ipk}+\sum_{r=1}^{m}\widetilde{\gamma}_{ipr} & =1.
\end{alignat*}
The basis $\left(e_{i}\right)_{1\text{\ensuremath{\le}}i\text{\ensuremath{\le}}n}\cup\left(\widetilde{e}_{j}\right)_{1\text{\ensuremath{\le}}j\text{\ensuremath{\le}}m}$
is called a \textit{gonosomal basis} of $A$.

\medskip{}

As established in \cite{RV-15}, gonosomal algebras can represent
algebraically all sex determination systems ($XY$, $WZ$, $X0$,
$Z0$ and $WXY$) but also provide an algebraic representation
of a wide variety of sex-related genetic phenomena. This variety
has been illustrated in \cite{RV-15} by giving several constructions
of gonosomal algebras and by motivating their relevance through
nearly twenty genetic examples of sex-linked inheritance mechanisms
observed in bisexual populations. However, certain sex-related
situations, such as sterility, cannot be accurately modelled
using gonosomal algebras. To remedy this limitation, we extend
in this article, the definition of gonosomal algebras by introducing
gonosomic algebras.

\medskip{}

The text is organised into six sections. After this introduction,
in section 2 we give an example that does not verify the definition
of a gonosomal algebra, this leads to the definition of the gonosomic
algebra, next we give two criteria to determine under what conditions
a gonosomic algebra does not admit a gonosomal basis. In section
3, a characterisation by bilinear maps and some constructions
of this type of algebras are given and applied to genetic examples.
In section 4 we define the notion of duplication of a gonosomic
algebra which makes it possible to infer the genotypic structure
of a sex-linked gene from the haplotype structure of this gene.
Section 5 is devoted to the study of a quadratic operator, called
a gonosomic operator. Each of these operators is associated with
a discrete-time dynamic system that models the change in the
population sizes of the sex-linked gene types between successive
generations. In section 6, starting from a gonosomic operator
$W$, we define an operator $V$ on the set of relative frequency
distributions, called normalised gonosomic operator. We show
that there is a one to one correspondence between a set of fixed
points of $W$ satisfying certain conditions and the set of all
fixed points of $V$, it is also shown that the various stability
notions of equilibrium points are preserved from $W$ to $V$.

\section{Definition and basic properties of gonosomic algebras}

\subsection{Definition and first properties of gonosomic algebras.}

\textcompwordmark{}

We extend the definition of gonosomal algebra, this extension
finds its source in the following example. 
\begin{example}
\emph{\label{exa:male_infertility-1}Genetic male infertility. }

Some cases of male infertility in the XY-system are due to genetic
abnormalities on the Y chromosome \cite{Bla)25}. If we denote
by $Y^{*}$ the $Y$ chromosome carrying an abnormality causing
infertility and by $\mu$ the incidence rate of this abnormality.
With this, the breeding $XX\times XY^{*}$ is sterile and results
of the crosses are \medskip{}

\hfill{}%
\begin{tabular}{c|c|c|}
 & $XY$ & $XY^{*}$\tabularnewline
\hline 
\hline 
$XX$ & $\frac{1}{2}XX,\frac{1-\mu}{2}XY,\frac{\mu}{2}XY^{*}$ & $0$\tabularnewline
\end{tabular}\hfill{}

\medskip{}

Algebraically, we consider the $\mathbb{\mathbb{K}}$-algebra
$A$ with basis $\left(e_{1},\widetilde{e}_{1},\widetilde{e}_{2}\right)$
defined by $e_{1}\tilde{e}_{1}=\tilde{e}_{1}e_{1}=\frac{1}{2}e_{1}+\frac{1-\mu}{2}\tilde{e}_{1}+\frac{\mu}{2}\tilde{e}_{2}$
and all other products are zero. It is clear that with these
products and the correspondences $e_{1}\leftrightarrow XX,\widetilde{e}_{1}\leftrightarrow XY,\widetilde{e}_{2}\leftrightarrow XY^{*}$,
we obtain the results of the crosses.

In particular, we can notice that the products $e_{1}\widetilde{e}_{2}=0$
do not satisfy condition \textit{iii}) of the gonosomal algebra
definition. We will see in remark \ref{rem:male-infertility}
that this algebra is not gonosomal because it does not admit
a gonosomal basis.

\medskip{}
\end{example}

In the following we extend the definition of gonosomal algebra
to take into account the situation described in the example given
above. To do this, condition \textit{iii}) is dropped from the
definition of a gonosomal algebra.

\medskip{}

\begin{defn}
\label{def:Gonosomic-Alg} Given a commutative field $\mathbb{K}$
with characteristic $\neq2$, a $\mathbb{K}$-algebra $A$ is
a \textit{gonosomic algebra} if it admits a basis $\left(e_{i}\right)_{i\in I}\cup\left(\widetilde{e}_{j}\right)_{j\in J}$
called \textit{gonosomic basis}, such that for all $i,j\in I$
and $p,q\in J$ we have: 
\begin{eqnarray*}
e_{i}e_{j} & = & 0,\\
\widetilde{e}_{p}\widetilde{e}_{q} & = & 0,\\
e_{i}\widetilde{e}_{p}\;=\;\widetilde{e}_{p}e_{i} & = & \sum_{k\in I}\gamma_{ipk}e_{k}+\sum_{r\in J}\widetilde{\gamma}_{ipr}\widetilde{e}_{r}.
\end{eqnarray*}
\end{defn}

When the index sets $I$ and $J$ are finite, $I=\left\{ 1,\ldots,n\right\} $,
$J=\left\{ 1,\ldots,m\right\} $ and the structure constants
verify $\sum_{k=1}^{n}\gamma_{ipk}+\sum_{r=1}^{m}\widetilde{\gamma}_{ipr}=1$
for all $i\in I$ and $p\in J$, the definition of a gonosomic
algebra corresponds to that of a gonosomal algebra of type $\left(n,m\right)$.
Therefore, we can say that gonosomic algebras extend gonosomal
algebras.

We will say that a gonosomic algebra it is not gonosomal if it
does not admit a gonosomal basis.\medskip{}

From a genetic point of view, in this definition the vectors
$\left(e_{i}\right)_{i\in I}$ (resp. $\left(\widetilde{e}_{j}\right)_{j\in J}$)
are interpreted as genetic types observed in females (resp. males),
the structure constant $\gamma_{ipk}$ (resp. $\widetilde{\gamma}_{ipr}$
) represents the female (resp.male) absolute or relative frequency
of type $e_{k}$ (resp. $\widetilde{e}_{r}$) in the progeny
of a female type $e_{i}$ with a male type $\widetilde{e}_{p}$.\medskip{}

We can give an equivalent definition of a gonosomic algebra.
\begin{prop}
Let $A$ be A commutative $\mathbb{K}$-algebra, $A$ is gonosomic
if and only if there are two subspaces $U$ and $V$ of $A$
such that $A=U\oplus V$ and $U^{2}=V^{2}=\left\{ 0\right\} $. 
\end{prop}

\begin{proof}
Let $\left(u_{i}\right)_{i\in I}$ and $\left(v_{j}\right)_{j\in J}$
respectively the basis of $U$ and $V$. According to the assumptions
we have $u_{i}u_{j}=0$ for $i,j\in I$ and $v_{i}v_{j}=0$ for
any $i,j\in J$. For all $i\in I$ and $j\in J$ we have $u_{i}v_{j}\in U+V$
thus $u_{i}v_{j}=\sum_{k\in I}\alpha_{ijk}u_{k}+\sum_{k\in J}\beta_{ijk}v_{k}$,
it follows that $\left(u_{i}\right)_{i\in I}\cup\left(v_{j}\right)_{j\in J}$
is a gonosomic basis of $U\oplus V$. 
\end{proof}
A gonosomic basis is not unique. 
\begin{prop}
Given $A$ a $\mathbb{K}$-gonosomic algebra and $\left(e_{i}\right)_{i\in I}\cup\left(\widetilde{e}_{j}\right)_{j\in J}$
a gonosomic basis of $A$. For any automorphisms $\varphi$ and
$\widetilde{\varphi}$ respectively of the vector spaces $\text{span}\left(\left(e_{i}\right)_{i\in I}\right)$
and $\text{span}\bigl(\left(\widetilde{e}_{j}\right)_{j\in J}\bigr)$,
the basis $\left(\varphi\left(e_{i}\right)\right)_{i\in I}\cup\left(\widetilde{\varphi}\left(\widetilde{e}_{j}\right)\right)_{j\in J}$
is gonosomic. 
\end{prop}

\begin{proof}
For every $i\in I$ and $p\in J$, let $\varphi\left(e_{i}\right)=\sum_{j\in I}\alpha_{ji}e_{j}$
and $\widetilde{\varphi}\left(\widetilde{e}_{p}\right)=\sum_{q\in I}\widetilde{\alpha}_{qp}\widetilde{e}_{q}$.
It is immediate that for all $i,j\in I$ and $p,q\in J$ we have
$\varphi\left(e_{i}\right)\varphi\left(e_{j}\right)=\widetilde{\varphi}\left(\widetilde{e}_{p}\right)\widetilde{\varphi}\left(\widetilde{e}_{q}\right)=0$.
Next we get $\varphi\left(e_{i}\right)\widetilde{\varphi}\left(\widetilde{e}_{p}\right)=\sum_{k\in I}\left(\sum_{j\in I,q\in J}\alpha_{ji}\gamma_{jqk}\widetilde{\alpha}_{qp}\right)e_{k}+\sum_{r\in J}\left(\sum_{j\in I,q\in J}\alpha_{ji}\widetilde{\gamma}_{jqk}\widetilde{\alpha}_{qp}\right)\widetilde{e}_{r}$. 
\end{proof}
\medskip{}

Throughout this paper for any integer $n\geq1$ we denote $N_{n}=\left\{ 1,\ldots,n\right\} $.

\medskip{}

A gonosomic algebra $A$ with finite gonosomic basis $\left(e_{i}\right)_{1\leq i\leq n}\cup\left(\widetilde{e}_{p}\right)_{1\leq p\leq m}$
is canonically equipped with the following linear form: 
\begin{equation}
\varpi:A\rightarrow\mathbb{\mathbb{K}},\quad\varpi\left(e_{i}\right)=\varpi\left(\widetilde{e}_{j}\right)=1.\label{eq:form_lin_pi_def-1}
\end{equation}

With this, for every $i\in N_{n}$ and $p\in N_{m}$ we get 
\begin{equation}
\varpi\left(e_{i}\widetilde{e}_{p}\right)=\sum_{k=1}^{n}\gamma_{ipk}+\sum_{r=1}^{m}\widetilde{\gamma}_{ipr}.\label{eq:omega(ei,ep)-1}
\end{equation}

\medskip{}

Let us introduce some notations that will be used in the proofs
of the following two theorems.
\begin{lem}
\label{lem:Notations}Let $A$ be a gonosomic $\mathbb{K}$-algebra
and $\left(e_{i}\right)_{1\leq i\leq n}\cup\left(\widetilde{e}_{p}\right)_{1\leq p\leq m}$
a gonosomic basis of $A$. If $A$ admits a gonosomal basis $\left(a_{i}\right)_{1\leq i\leq n}\cup\left(\widetilde{a}_{j}\right)_{1\leq j\leq m}$,
then for any $i\in N_{n}$ and $p\in N_{m}$ we have $e_{i}=\sum_{k=1}^{n}\alpha_{ik}a_{k}+\sum_{r=1}^{m}\beta_{ir}\widetilde{a}_{r}$
and $\widetilde{e}_{p}=\sum_{k=1}^{n}\widetilde{\alpha}_{pk}a_{k}+\sum_{r=1}^{m}\widetilde{\beta}_{pr}\widetilde{a}_{r}$.

For any $i\in N_{n}$ and $p\in N_{m}$ we have
\begin{equation}
\varpi\left(e_{i}\right)=\alpha_{i}+\beta_{i},\qquad\varpi\left(\widetilde{e}_{p}\right)=\widetilde{\alpha}_{p}+\widetilde{\beta}_{p}.\label{eq:omega(e_i,etild_p)}
\end{equation}
\begin{equation}
\varpi\left(e_{i}\widetilde{e}_{p}\right)=\alpha_{i}\widetilde{\beta}_{p}+\beta_{i}\widetilde{\alpha}_{p}.\label{eq:omega(e_i_X_etild_p)}
\end{equation}
\begin{equation}
\alpha_{i}\beta_{j}+\beta_{i}\alpha_{j}=0,\qquad\widetilde{\alpha}_{i}\widetilde{\beta}_{j}+\widetilde{\beta}_{i}\widetilde{\alpha}_{j}=0.\label{eq:omega_0}
\end{equation}

Where $\alpha_{i}=\sum_{k=1}^{n}\alpha_{ik}$, $\beta_{i}=\sum_{r=1}^{m}\beta_{ir}$,
as well as $\widetilde{\alpha}_{p}=\sum_{k=1}^{n}\widetilde{\alpha}_{pk}$
and $\widetilde{\beta}_{p}=\sum_{r=1}^{m}\widetilde{\beta}_{pr}$. 
\end{lem}

\begin{proof}
For any $i\in N_{n}$ and $p\in N_{m}$ we have $\varpi\left(a_{i}\right)=\varpi\left(\widetilde{a}_{p}\right)=1$,
so by applying the form $\varpi$ to the expressions giving $e_{i}$
and $\widetilde{e}_{p}$ according to the basis $\left(a_{i}\right)_{1\leq i\leq n}\cup\left(\widetilde{a}_{j}\right)_{1\leq j\leq m}$
we get (\ref{eq:omega(e_i,etild_p)}). 

Next, we have
\[
e_{i}\widetilde{e}_{p}=\sum_{k,r=1}^{n,m}\alpha_{ik}\widetilde{\beta}_{pr}a_{k}\widetilde{a}_{r}+\sum_{k,r=1}^{n,m}\beta_{ir}\widetilde{\alpha}_{pk}a_{r}\widetilde{a}_{k},
\]
after applying the form $\varpi$ to this relation and using
the fact that $\varpi\left(a_{k}\widetilde{a}_{r}\right)=1$
for all $k\in N_{n}$, $r\in N_{m}$, we get (\ref{eq:omega(e_i_X_etild_p)}).
In a similar way, from $e_{i}e_{j}=0$ and $\widetilde{e}_{i}\widetilde{e}_{j}=0$
we get (\ref{eq:omega_0}).
\end{proof}
\medskip{}

In what follows, we give necessary conditions under which a gonosomic
algebra does not admit a gonosomal basis. It follows that there
are gonosomic algebras that are not gonosomal, in other words
the class of gonosomal algebras is strictly included in that
of gonosomic algebras.
\begin{thm}
\label{thm:not_gonosomal_1}Let $A$ be a gonosomic $\mathbb{K}$-algebra
and $\left(e_{i}\right)_{1\leq i\leq n}\cup\left(\widetilde{e}_{p}\right)_{1\leq p\leq m}$
a gonosomic basis of $A$ with $n,m\geq2$. If it exists $i,j\in N_{n}$
and $p\in N_{m}$ such that $\varpi\left(e_{i}\widetilde{e}_{p}\right)=0$
and $\varpi\left(e_{j}\widetilde{e}_{p}\right)\neq0$, then the
algebra $A$ does not admit a gonosomal basis. 
\end{thm}

\begin{proof}
Suppose that $A$ has a gonosomal basis $\left(a_{i}\right)_{1\leq i\leq n}\cup\left(\widetilde{a}_{j}\right)_{1\leq j\leq m}$.
Using the notations of the lemma \ref{lem:Notations}, from $e_{i}^{2}=0$
and (\ref{eq:omega_0}) we get $\alpha_{i}\beta_{i}=0$, we have
not $\alpha_{i}=\beta_{i}=0$ otherwise according to (\ref{eq:omega(e_i,etild_p)})
we would have $\varpi\left(e_{i}\right)=0$, thus $\alpha_{i}\neq0$
and $\beta_{i}=0$ or $\alpha_{i}=0$ and $\beta_{i}\neq0$.
Suppose that $\alpha_{i}\neq0$ and $\beta_{i}=0$, then for
any $k\in N_{n}$ according to (\ref{eq:omega_0}) we have $\alpha_{i}\beta_{k}=0$
therefore $\beta_{k}=0$ for any $k\in N_{n}$. According to
(\ref{eq:omega(e_i_X_etild_p)}) we have $\varpi\left(e_{j}\widetilde{e}_{p}\right)=\alpha_{j}\widetilde{\beta}_{p}+\beta_{j}\widetilde{\alpha}_{p}$
with $\beta_{j}=0$ therefore $\alpha_{j}\widetilde{\beta}_{p}\neq0$
which implies $\widetilde{\beta}_{p}\neq0$. But we have $\varpi\left(e_{i}\widetilde{e}_{p}\right)=0$
with $\beta_{i}=0$ and according to (\ref{eq:omega(e_i_X_etild_p)})
we get $\alpha_{i}\widetilde{\beta}_{p}=0$ thus $\alpha_{i}=0$,
contradiction.

Assuming that $\alpha_{i}=0$ and $\beta_{i}\neq0$, the same
contradiction is obtained by exchanging letters $\alpha$ and
$\beta$.
\end{proof}
\medskip{}

It follows from theorem \ref{thm:not_gonosomal_1} that a necessary
condition for a gonosomic algebra not to be gonosomal is that
some crosses are totally sterile. 
\begin{rem}
\label{rem:male-infertility}In example \ref{exa:male_infertility-1},
because $\varpi\left(e_{1}\widetilde{e}_{1}\right)=1$ and $\varpi\left(e_{1}\widetilde{e}_{2}\right)=0$,
we deduce from this theorem that the gonosomic algebra given
in this example is not gonosomal.
\end{rem}

\begin{thm}
\label{thm:not_gonosomal_2}Let $A$ be a gonosomic $\mathbb{K}$-algebra
and $\left(e_{i}\right)_{1\leq i\leq n}\cup\left(\widetilde{e}_{p}\right)_{1\leq p\leq m}$
a gonosomic basis of $A$ such that $n,m\geq2$ and $\varpi\left(e_{k}\widetilde{e}_{r}\right)\neq0$
for any $k\in N_{n}$, $r\in N_{m}$. If it exists $i,j\in N_{n}$
and $p,q\in N_{m}$ such that $\det\left(\begin{array}{cc}
\varpi\left(e_{i}\widetilde{e}_{p}\right) & \varpi\left(e_{i}\widetilde{e}_{q}\right)\\
\varpi\left(e_{j}\widetilde{e}_{p}\right) & \varpi\left(e_{j}\widetilde{e}_{q}\right)
\end{array}\right)\neq0$, then the algebra $A$ does not admit a gonosomal basis. 
\end{thm}

\begin{proof}
To simplify the calculations, we set $\varpi\left(e_{i}\widetilde{e}_{p}\right)=\alpha$,
$\varpi\left(e_{j}\widetilde{e}_{q}\right)=\beta$, $\varpi\left(e_{i}\widetilde{e}_{q}\right)=\lambda$
and $\varpi\left(e_{j}\widetilde{e}_{p}\right)=\mu$.

If we assume that $A$ admits a gonosomal basis $\left(a_{i}\right)_{1\leq i\leq n}\cup\left(\widetilde{a}_{j}\right)_{1\leq j\leq m}$.
Using the lemma \ref{lem:Notations}, with (\ref{eq:omega_0})
we get $\varpi\left(e_{i}^{2}\right)=2\alpha_{i}\beta_{i}=0$
and with (\ref{eq:omega(ei,ep)-1}) and (\ref{eq:omega(e_i,etild_p)})
we deduce that $\left(\alpha_{i},\beta_{i}\right)\neq\left(0,0\right)$.
Suppose $\alpha_{i}\neq0$ and $\beta_{i}=0$, with (\ref{eq:omega(e_i_X_etild_p)})
we get $\varpi\left(e_{i}\widetilde{e}_{p}\right)=\alpha_{i}\widetilde{\beta}_{p}+\beta_{i}\widetilde{\alpha}_{p}=\alpha$
hence $\alpha_{i}\widetilde{\beta}_{p}=\alpha$ from this we
deduce that $\widetilde{\beta}_{p}\neq0$. In a similar way we
have $\varpi\left(e_{i}\widetilde{e}_{q}\right)=\alpha_{i}\widetilde{\beta}_{q}=\lambda$.
It follows that $\frac{\widetilde{\beta}_{q}}{\widetilde{\beta}_{p}}=\frac{\lambda}{\alpha}$
or $\widetilde{\beta}_{q}=\frac{\lambda}{\alpha}\widetilde{\beta}_{p}\;\left(*\right)$. 

From $\varpi\left(e_{i}e_{j}\right)=0$ according to (\ref{eq:omega_0})
and $\beta_{i}=0$ we get $\alpha_{i}\beta_{j}=0$, because $\alpha_{i}\neq0$
we have $\beta_{j}=0$, with this and (\ref{eq:omega(e_i_X_etild_p)})
we get $\varpi\left(e_{j}\widetilde{e}_{p}\right)=\alpha_{j}\widetilde{\beta}_{p}=\mu$
and $\varpi\left(e_{j}\widetilde{e}_{q}\right)=\alpha_{j}\widetilde{\beta}_{q}=\beta$,
because $\beta\neq0$ we have $\alpha_{j}\neq0$ and $\widetilde{\beta}_{q}\neq0$,
therefore $\frac{\widetilde{\beta}_{p}}{\widetilde{\beta}_{q}}=\frac{\mu}{\beta}$
or $\widetilde{\beta}_{p}=\frac{\mu}{\beta}\widetilde{\beta}_{q}\;\left(**\right)$. 

From $\left(*\right)$ and $\left(**\right)$ we deduce that
$\widetilde{\beta}_{q}=\frac{\lambda}{\alpha}\frac{\mu}{\beta}\widetilde{\beta}_{q}$
with $\widetilde{\beta}_{q}\neq0$, thus we have $\alpha\beta-\lambda\mu=0$,
in other words $\det\left(\begin{array}{cc}
\varpi\left(e_{i}\widetilde{e}_{p}\right) & \varpi\left(e_{i}\widetilde{e}_{q}\right)\\
\varpi\left(e_{j}\widetilde{e}_{p}\right) & \varpi\left(e_{j}\widetilde{e}_{q}\right)
\end{array}\right)=0$.

By a similar reasoning, the same result is obtained if we assume
$\alpha_{i}=0$ and $\beta_{i}\neq0$. 
\end{proof}
The following examples provide genetic applications of these
two theorems.

\smallskip{}

\begin{example}
\textit{\label{exa:bidirectional}Bidirectional cytoplasmic incompatibility.}

{\small\textit{\emph{Cytoplasmic incompatibility is}}}{\small{}
used in biological pest control. The bidirectional cytoplasmic
incompatibility is a mating incompatibility caused by parasites
that reside in the cytoplasm of germ cells (sperm and/or eggs)
\cite{Tur-22}. Bidirectional incompatibility is observed when
there are two types of parasites, the crossing between two organisms
infected with different types of parasites is sterile otherwise
it is fertile.}\vspace{2mm}

Consider the case of two types of parasites denoted 1 and 2 and
assuming that the cross between two organisms infected by different
parasites is sterile. Algebraically, let $f_{1},f_{2}$ (resp.
$m_{1},m_{2}$) be the types of infected females (resp. males),
$A$ the algebra defined on the basis $\left(f_{1},f_{2}.m_{1},m_{2}\right)$
by
\[
f_{1}m_{1}=m_{1}f_{1}=\frac{1}{2}f_{1}+\frac{1}{2}m_{1},\;f_{2}m_{2}=m_{2}f_{2}=\frac{1}{2}f_{2}+\frac{1}{2}m_{2},
\]
the other products being zero.

As we have $\varpi\left(f_{1}m_{2}\right)=0$, $\varpi\left(f_{1}m_{2}\right)=0$
and  $\varpi\left(f_{1}m_{1}\right)=0$, according to the theorem
\ref{thm:not_gonosomal_1} the algebra $A$ is not gonosomal. 
\end{example}

\smallskip{}

\begin{example}
\emph{\label{exa:Hybrid-dysgenesis}Hybrid dysgenesis in Drosophila
melanogaster.}

In the species \emph{D. melanogaster} there are two strains:
$M$ and $P$. When these strains are crossed at a temperature
of 28-29°C, the following results are observed:\medskip{}

\hspace{3cm}%
\begin{tabular}{c|c|c|}
 & $M\male$ & $P\male$\tabularnewline
\hline 
\begin{tabular}{c}
\tabularnewline
\tabularnewline
\end{tabular}$M\female$ & $\frac{1}{2}M\female,\frac{1}{2}M\male$ & $0$\tabularnewline
\hline 
\begin{tabular}{c}
\tabularnewline
\tabularnewline
\end{tabular}$P\female$ & $\frac{1}{2}P\female,\frac{1}{2}P\male$ & $\frac{1}{2}P\female,\frac{1}{2}P\male$\tabularnewline
\hline 
\end{tabular}\medskip{}

Algebraically, we consider the algebra with basis $\left(e_{i},\widetilde{e}_{i}\right)_{1\leq i\leq2}$
defined by: $e_{1}\widetilde{e}_{1}=\frac{1}{2}e_{1}+\frac{1}{2}\widetilde{e}_{1}$;
$e_{2}\widetilde{e}_{2}=\frac{1}{2}e_{2}+\frac{1}{2}\widetilde{e}_{2}$;
$e_{1}\widetilde{e}_{2}=0$; $e_{2}\widetilde{e}_{1}=\frac{1}{2}e_{2}+\frac{1}{2}\widetilde{e}_{2}$,
the other products being zeros. By using in these products the
following correspondences $e_{1}\leftrightarrow M\female$; $e_{2}\leftrightarrow P\female$;
$\widetilde{e}_{1}\leftrightarrow M\male$; $\widetilde{e}_{2}\leftrightarrow P\male$
we find all the results of the crosses given in the table above.
Furthermore, as $\varpi\left(e_{1}\widetilde{e}_{2}\right)=0$
and $\varpi\left(e_{1}\widetilde{e}_{1}\right)=1$, according
to the theorem \ref{thm:not_gonosomal_1} this algebra is not
gonosomal. 
\end{example}

\medskip{}

\begin{example}
\textit{\label{exa:partial infertility}Bisexual population with
partial sterility}.

In the XY-system, chromosome aberrations reduce the number of
viable gametes and hence reduce fertility which results in partial
sterility \cite{Luk-20}. 

If we denote by $X^{*}$ and $Y^{*}$ respectively a female and
a male having reduced fertility. If $0\leq\mu,\nu,\tau\leq1$
denote sterility rates, we have the crossbreeding table:
\begin{alignat*}{2}
X\times Y & \rightarrowtail\tfrac{1}{2}X,\tfrac{1}{2}Y; & \hspace{1cm}X\times Y^{*} & \rightarrowtail\tfrac{1-\mu}{2}X,\tfrac{1-\mu}{2}Y^{*};\\
X^{*}\times Y & \rightarrowtail\tfrac{1-\nu}{2}X^{*},\tfrac{1-\nu}{2}Y; & X^{*}\times Y^{*} & \rightarrowtail\tfrac{1-\tau}{2}X^{*},\tfrac{1-\tau}{2}Y^{*};
\end{alignat*}
In particular, if $\mu=0$, the cross $X\times Y^{*}$ is fertile,
and if $\mu=1$, the cross $X\times Y^{*}$ is completely sterile. 

Given the correspondences $e_{1}\leftrightarrow X,e_{2}\leftrightarrow X^{*},\widetilde{e}_{1}\leftrightarrow Y,\widetilde{e}_{2}\leftrightarrow Y^{*}$,
we introduce the commutative $\mathbb{\mathbb{K}}$-algebra $A$
with basis $\left(e_{1},e_{2},\widetilde{e}_{1},\widetilde{e}_{2}\right)$
defined by 
\begin{alignat*}{2}
e_{1}\tilde{e}_{1} & =\tfrac{1}{2}e_{1}+\tfrac{1}{2}\tilde{e}_{1}; & \hspace{1cm}e_{1}\tilde{e}_{2} & =\tfrac{1-\mu}{2}e_{1}+\tfrac{1-\mu}{2}\tilde{e}_{2};\\
e_{2}\tilde{e}_{1} & =\tfrac{1-\nu}{2}e_{2}+\tfrac{1-\nu}{2}\tilde{e}_{1}; & e_{2}\tilde{e}_{2} & =\tfrac{1-\tau}{2}e_{2}+\tfrac{1-\tau}{2}\tilde{e}_{2}.
\end{alignat*}
This algebra $A$ is gonosomic and if the rates $\mu,\nu,\tau$
verify $\left(1-\mu\right)\left(1-\nu\right)\neq\left(1-\tau\right)$,
then according to theorem \ref{thm:not_gonosomal_2} the algebra
$A$ is not gonosomal.
\end{example}

\medskip{}

\section{Characterisation by bilinear maps and some construction of gonosomic
algebras.}

The following result gives a characterisation of gonosomic algebras
using bilinear maps. 
\begin{thm}
\label{thm:Caract}A $\mathbb{K}$-algebra $A$ is gonosomic
if and only if $A$ is isomorphic to the $\mathbb{K}$-algebra
$B\times\widetilde{B}$ defined by 
\[
\left(x,y\right)\left(x',y'\right)=\left(b\left(x,y'\right)+b\left(x',y\right),\widetilde{b}\left(x,y'\right)+\widetilde{b}\left(x',y\right)\right),
\]
where $B$, $\widetilde{B}$ are two $\mathbb{K}$-vector spaces
; $b:B\times\widetilde{B}\rightarrow B$ and $\widetilde{b}:B\times\widetilde{B}\rightarrow\widetilde{B}$
are two bilinear maps. 
\end{thm}

\begin{proof}
Let $\left(a_{i}\right)_{i\in I}$ and $\left(\widetilde{a}_{p}\right)_{p\in J}$
be respectively a basis of $B$ and $\widetilde{B}$, given $b\left(a_{i},\widetilde{a}_{p}\right)=\sum_{k\in I}\gamma_{ipk}a_{k}$
and $\widetilde{b}\left(a_{i},\widetilde{a}_{p}\right)=\sum_{q\in J}\widetilde{\gamma}_{ipq}\widetilde{a}_{q}$.
For all $i,j\in I$ and $p,q\in J$, from the definition we have
$\left(a_{i},0\right)\left(a_{j},0\right)=0$, $\left(0,\widetilde{a}_{p}\right)\left(0,\widetilde{a}_{q}\right)=0$
and $\left(a_{i},0\right)\left(0,\widetilde{a}_{p}\right)=\left(b\left(a_{i},\widetilde{a}_{p}\right),\widetilde{b}\left(a_{i},\widetilde{a}_{p}\right)\right)=\sum_{k\in I}\gamma_{ipk}\left(a_{k},0\right)+\sum_{q\in J}\widetilde{\gamma}_{ipq}\left(0,\widetilde{a}_{q}\right)$.
It follows from this that if $A$ is a gonosomic algebra with
a gonosomic basis $\left(e_{i}\right)_{i\in I}\cup\left(\widetilde{e}_{p}\right)_{p\in J}$,
the linear map $\left(a_{i},0\right)\mapsto e_{i}$ , $\left(0,\widetilde{a}_{p}\right)\mapsto\widetilde{e}_{p}$
is an algebra isomorphism between $B\times\widetilde{B}$ and
$A$. 
\end{proof}
The following two examples show that this result is very useful
for easily formulating algebraic models of sex-linked genetic
situations.
\begin{example}
\textit{Haplodiploidy}

{\small Haplodiploidy is a sex-determination system of type XX-X0
which controls the sex in bees, ants and wasps. In this system
females develop from fertilised eggs and are diploid and male
come from unfertilised eggs and are haploid. These insects live
in colonies that have only one fully fertile female called the
queen, which is fertilised only once by several males. After
mating, fertilised females store sperm in a spermatheca. The
fertilised female controls the release of sperm: when an egg
descends into the oviduct, she releases sperm to fertilise it.
Through this mechanism, a queen can modify the sex ratio in the
colony.}{\small\par}

\vspace{2mm}

Consider a colony of haplodiploid organisms, let $a_{1},\ldots,a_{n}$
be the haplotypes present in this colony and $a,b\in\left\{ a_{1},\ldots,a_{n}\right\} $.
The queen's genotype is denoted by $a\otimes b$, the genotypes
of the females by $a\otimes a_{i}$ and $b\otimes a_{i}$, and
the genotypes of the males by $a$, $b$, $a_{i}$. Let $B,\widetilde{B}$
be vector spaces, with basis $\left(a\otimes b,a\otimes a_{i},b\otimes a_{i}\right)_{1\leq i\leq n}$
a basis of $B$ and $\left(a_{i}\right)_{1\leq i\leq n}$ a basis
of $\widetilde{B}$. We denote by $\theta$ the frequency of
females in the colony and we define on $B\times\widetilde{B}$
the symmetric bilinear maps $b$ and $\widetilde{b}$ :\vspace{2mm}

\hfill{}%
\begin{tabular}{ll}
$b\left(a\otimes b,a_{i}\right)=\frac{\theta}{2}a\otimes a_{i}+\frac{\theta}{2}b\otimes a_{i}$, & $b\left(a\otimes b,a\otimes a_{i}\right)=b\left(a\otimes b,b\otimes a_{i}\right)=0$,\tabularnewline
$\widetilde{b}\left(a\otimes b,a_{i}\right)=\frac{1-\theta}{2}a+\frac{1-\theta}{2}b$, & $\widetilde{b}\left(a\otimes a_{j},a_{i}\right)=\widetilde{b}\left(b\otimes a_{j},a_{i}\right)=0$.\tabularnewline
\end{tabular}\hfill{}\vspace{2mm}

According to theorem \ref{thm:Caract}, the algebra $B\times\widetilde{B}$
is gonosomic and give the frequency distribution of haplotypes
in a clutch.
\end{example}

\medskip{}

\begin{example}
\textit{Sex determination in }\emph{Caenorhabditis elegans.}

{\small The nematode }{\small\emph{Caenorhabditis elegans}}{\small{}
is a small, transparent worm about one millimetre long. }{\small\emph{C.
elegans}}{\small{} has a special reproductive system because it
has two sexes: hermaphrodite and male, therefore it can reproduce
by self-fertilisation or by cross-fertilisation. In }{\small\emph{C.
elegans}}{\small , sex is determined by the X0 system: hermaphrodite
(XX) and male (X0). When a hermaphrodite reproduces by self-fertilisation,
due to the non-disjunction of the X chromosomes, it produces
a proportion of 0.2\% males and when it crosses with a male it
gives birth to 50\% males. However, mutations in the }{\small\emph{him}}{\small{}
(high incidence of males) gene increase the proportion of males
by a factor ranging from 10 to 150.}\vspace{2mm}

Denoting by $\left(XX\right)_{\tau}$ the hermaphroditic individuals
that give birth to a proportion $\tau$ of males during self-fertilisation,
we therefore have
\begin{align*}
\left(XX\right)_{\tau}\times\left(XX\right)_{\sigma} & \rightarrowtail\begin{cases}
\left(1-\tau\right)\left(XX\right)_{\tau}+\tau X0 & \text{if }\tau=\sigma\\
0 & \text{if }\tau\neq\sigma
\end{cases},\\
\left(XX\right)_{\tau}\times X0 & \rightarrowtail\tfrac{1}{2}\left(XX\right)_{\tau}+\tfrac{1}{2}X0.
\end{align*}
Algebraically, we consider the $\mathbb{R}$-vector spaces $B$
and $\widetilde{B}$ with respective bases $\left\{ e_{\tau};\tau\in I\right\} $
and $\left\{ \widetilde{e},\widetilde{e}_{\tau};\tau\in I\right\} $
where $I=\left[0,1\right]$. Let $0\leq\theta\leq1$ be the reproduction
rate by self-fertilisation. We define on $B\times\widetilde{B}$
the symmetric bilinear maps $b$ and $\widetilde{b}$:\vspace{2mm}

\hfill{}%
\begin{tabular}{lll}
$b\left(e_{\tau},\widetilde{e}_{\tau}\right)=\theta\left(1-\tau\right)e_{\tau}$, & $b\left(e_{\tau},\widetilde{e}\right)=\frac{1-\theta}{2}e_{\tau}$, & $b\left(e_{\tau},\widetilde{e}_{\sigma}\right)=0,\;\left(\tau\neq\sigma\right)$\tabularnewline
$\widetilde{b}\left(e_{\tau},\widetilde{e}_{\tau}\right)=\theta\tau\widetilde{e}_{\tau}$, & $\widetilde{b}\left(e_{\tau},\widetilde{e}\right)=\frac{1-\theta}{2}\widetilde{e}$, & $\widetilde{b}\left(e_{\tau},\widetilde{e}_{\sigma}\right)=0,\;\left(\tau\neq\sigma\right)$.\tabularnewline
\end{tabular}\hfill{}\vspace{2mm}

According to theorem \ref{thm:Caract}, the algebra $B\times\widetilde{B}$
is gonosomic. Let be $x=\sum_{\tau\in I}\alpha_{\tau}e_{\tau}+\widetilde{\alpha}\widetilde{e}+\sum_{\tau\in I}\widetilde{\alpha}_{\tau}\widetilde{e}_{\tau}$
and $y=\sum_{\tau\in I}\beta_{\tau}e_{\tau}+\widetilde{\beta}\widetilde{e}+\sum_{\tau\in I}\widetilde{\beta}_{\tau}\widetilde{e}_{\tau}$
the frequency distributions (absolute or relative) of the sexes
in two populations, then the product
\begin{align*}
xy & =\sum_{\tau\in I}\left(\theta\left(1-\tau\right)\left(\alpha_{\tau}\widetilde{\beta}_{\tau}+\widetilde{\alpha}_{\tau}\beta_{\tau}\right)+\frac{1-\theta}{2}\left(\alpha_{\tau}\widetilde{\beta}+\beta_{\tau}\widetilde{\alpha}\right)\right)e_{\tau}\\
 & +\frac{1-\theta}{2}\sum_{\tau\in I}\left(\alpha_{\tau}\widetilde{\beta}_{\tau}+\widetilde{\alpha}_{\tau}\beta_{\tau}\right)\widetilde{e}+\theta\sum_{\tau\in I}\tau\left(\alpha_{\tau}\widetilde{\beta}_{\tau}+\widetilde{\alpha}_{\tau}\beta_{\tau}\right)\widetilde{e}_{\tau}
\end{align*}

gives the distribution of \emph{C. elegans} types obtained after
crossing the two populations $x$ and $y$.
\end{example}

\begin{prop}
\label{prop:mixture}Construction by mixture of gonosomic algebras.

Let $A$ a $\mathbb{K}$-vector space provided with two algebra
laws $\mu_{1},\mu_{2}:A\times A\rightarrow A$, if the algebras
$\left(A,\mu_{1}\right)$ and $\left(A,\mu_{2}\right)$ are gonosomic
with gonosomic basis $\mathcal{B}=\left(e_{i}\right)_{i\in I}\cup\left(\widetilde{e}_{j}\right)_{j\in J}$,
then the space $A$ with the product 
\[
xy=\mu_{1}\left(x,y\right)+\mu_{2}\left(x,y\right)
\]
is a gonosomic algebra with $\mathcal{B}$ as gonosomic basis. 
\end{prop}

\begin{proof}
If for $r=1,2$ we have $\mu_{r}\left(e_{i},e_{j}\right)=0$,
then we get $e_{i}e_{j}=\theta_{1}\mu_{1}\left(e_{i},e_{j}\right)+\theta_{2}\mu_{2}\left(e_{i},e_{j}\right)=0$
for all $i,j\in I$. Similarly from $\mu_{r}\left(\widetilde{e}_{i},\widetilde{e}_{j}\right)=0$
we deduce $\widetilde{e}_{i}\widetilde{e}_{j}=0$ for all $i,j\in J$.
And if $\mu_{r}\left(e_{i},\widetilde{e}_{j}\right)=\sum_{k\in I}\gamma_{ijk}^{\left(r\right)}e_{k}+\sum_{p\in J}\widetilde{\gamma}_{ijp}^{\left(r\right)}\widetilde{e}_{p}$
for all $i\in I$, $j\in J$, thus we get $e_{i}\widetilde{e}_{j}=\sum_{k\in I}\left(\theta_{1}\gamma_{ijk}^{\left(1\right)}+\theta_{2}\gamma_{ijk}^{\left(2\right)}\right)e_{k}+\sum_{p\in J}\left(\theta_{1}\widetilde{\gamma}_{ijk}^{\left(1\right)}+\theta_{2}\widetilde{\gamma}_{ijk}^{\left(2\right)}\right)\widetilde{e}_{k}$. 
\end{proof}
Here are two genetic examples using this result.

\medskip{}

\begin{example}
\textit{Thermal infertility}

{\small Physiology and behaviour of insects depend on the ambient
temperature. For a given species, the thermal temperature lies
within a range. Outside this temperature range, the physiology
and behaviour of an insect are disrupted, for example, changes
in sex ratio, infertility and a decrease in the number of eggs
can be observed.} \vspace{2mm}

Let us note respectively $f$ and $m$ the female and male genetic
types, $\left[t_{1},t_{2}\right]$ the optimum temperature range.
Let $A$ be the vector space spanned by the basis $\left(f,m\right)$,
equipped with the two laws of algebras $\mu_{1}\left(f,m\right)=\frac{1}{2}f+\frac{1}{2}m$,
$\mu_{1}\left(f,f\right)=\mu_{1}\left(m,m\right)=0$ and $\mu_{2}\left(f,m\right)=\mu_{2}\left(m,f\right)=\sigma\left(\frac{1}{1+\rho}f+\frac{\rho}{1+\rho}m\right)$,
$\mu_{2}\left(f,f\right)=\mu_{2}\left(m,m\right)=0$ with $0\leq\rho,\sigma\leq1$
and where $\sigma$ is the infertility rate observed when the
temperature of the environment is not within the range $\left[t_{1},t_{2}\right]$.
If we denote by $\theta$ the probability that the temperature
of the environment is between $t_{1}$ and $t_{2}$ , so the
product $fm=\theta\mu_{1}\left(f,m\right)+\left(1-\theta\right)\mu_{2}\left(f,m\right)$
gives the frequency distribution of sexes according to the temperature
range. 
\end{example}

\medskip{}

\begin{example}
\textit{Xenoparity in} \textit{Mussor ibericus}

{\small Recently, an astonishing mode of reproduction was discovered
in }{\small\textit{Messor ibericus}}{\small{} ants, for which the
term }{\small\textit{xenoparity}}{\small{} was coined \cite{Juv-25}.
Like all ants, }{\small\textit{M. ibericus}}{\small{} is haplodiploid,
but in this species, queens give birth to male ants of two different
species: males of their own species to produce queens, and males
of another species (}{\small\textit{Messor structor}}{\small )
to produce workers. The queen gives birth to }{\small\textit{M.
ibericus}}{\small{} males from unfertilised eggs or to }{\small\textit{M.
ibericus}}{\small{} queens if the eggs are fertilised by }{\small\textit{M.
ibericus}}{\small{} males. To produce worker bees, the queen uses
the sperm from }{\small\textit{M. structor}}{\small{} males stored
in her spermatheca. When a spermatozoon from }{\small\textit{M.
structor}}{\small{} fertilises an egg, there are two possibilities:
either a diploid female is produced, which is therefore a hybrid
of }{\small\textit{M. ibericus/structor}}{\small , or the queen's
nuclear genome is eliminated, in which case a male of the species
}{\small\textit{M. structor}}{\small{} is produced.}{\small\par}

\vspace{2mm}

Algebraically, consider a colony of {\small\textit{M. ibericus}}
ants, let $a_{1},\ldots,a_{n}$ and $b_{1},\ldots,b_{m}$ be
the haplotypes present in this colony. The {\small\textit{M.
ibericus}} queen's genotype is denoted by $a\otimes a'$ where
$a,a'\in\left\{ a_{1},\ldots,a_{n}\right\} $, the genotypes
of {\small\textit{M. ibericus}} females (queens) by $a\otimes a_{i}$
and $a'\otimes a_{i}$, and the genotypes of the {\small\textit{M.
ibericus}} males by $a_{i}$. The genotypes of \textit{M. structor}
males are denoted by $b_{i}$ and of hybrid {\small\textit{M.
ibericus/structor}} females by $a\otimes b_{i}$ and $a'\otimes b_{i}$.
let $A$ be the vector space and $\left(a\otimes a',a\otimes a_{i},a'\otimes a_{i},a_{i}\right)_{1\leq i\leq n}\cup\left(a\otimes b_{j},a'\otimes b_{j},b_{j}\right)_{1\leq j\leq m}$
a basis of $A$. To apply proposition \ref{prop:mixture}, let
$0\leq\delta,\eta\leq1$ be respectively the frequencies of {\small\textit{M.
ibericus}} and \textit{M. structor} males and $0\leq\theta\leq1$
the frequency of mating between the queen {\small\textit{M. ibericus}}
and males {\small\textit{M. structor}}, we define on $A\times A$
multiplications $\mu_{1}$ and $\mu_{2}$ by
\begin{align*}
\mu_{1}\left(a\otimes a',a_{p}\right) & =\left(1-\theta\right)\left(\frac{1-\delta}{2}\left(a\otimes a_{p}+a'\otimes a_{p}\right)+\delta a_{p}\right),\\
\mu_{2}\left(a\otimes a',b_{j}\right) & =\theta\left(\frac{1-\eta}{2}\left(a\otimes b_{j}+a'\otimes b_{j}\right)+\eta b_{j}\right).
\end{align*}
\end{example}

\medskip{}

\begin{prop}
\label{prop:Prod_tens}Given $A$ a gonosomic $\mathbb{K}$-algebra
and $A_{1},\ldots,A_{n}$ not necessary commutative $\mathbb{K}$-algebras.
Let $G=A\otimes A_{1}\otimes\cdots\otimes A_{n}$ and $\Psi:G\rightarrow G$
a linear map. Then the vector space $G$ equipped with the law
\[
\left(x\otimes x_{1}\otimes\cdots\otimes x_{n}\right)\left(y\otimes x'_{1}\otimes\cdots\otimes x'_{n}\right)=\frac{1}{2^{n}}\Psi\left(xy\otimes\bigotimes_{i=1}^{n}\left(x_{i}x'_{i}+x'_{i}x_{i}\right)\right)
\]
is a gonosomic algebra. 
\end{prop}

\begin{proof}
By induction on the integer $n$. For $n=1$, let $\left(e_{i}\right)_{i\in I}\cup\left(\widetilde{e}_{j}\right)_{j\in J}$
be a gonosomic basis of $A$ with $e_{i}\widetilde{e}_{j}=\sum_{k\in I}\gamma_{ijk}e_{k}+\sum_{p\in J}\widetilde{\gamma}_{ijp}\widetilde{e}_{p}$
and $\left(a_{u}\right)_{u\in U}$ a basis of $A_{1}$ with $a_{u}a_{v}=\sum_{w\in U}\lambda_{uvw}a_{w}$.
With this, for any $i,j\in I$ and $u,v\in U$ we get $\left(e_{i}\otimes a_{u}\right)\left(e_{j}\otimes a_{v}\right)=0\otimes a_{u}a_{v}=0$
and for all $i,j\in J$ and $u,v\in U$ we get $\left(\widetilde{e}_{i}\otimes a_{u}\right)\left(\widetilde{e}_{j}\otimes a_{v}\right)=0$.
Next for $i\in I$, $j\in J$ and $u,v\in U$ we get
\begin{align*}
\left(e_{i}\otimes a_{u}\right)\left(\widetilde{e}_{j}\otimes a_{v}\right) & =\frac{1}{2}\sum_{\left(k,w\right)\in I\times U}\gamma_{ijk}\left(\lambda_{uvw}+\lambda_{vuw}\right)e_{k}\otimes a_{w}+\\
 & \frac{1}{2}\sum_{\left(p,w\right)\in J\times U}\widetilde{\gamma}_{ijp}\left(\lambda_{uvw}+\lambda_{vuw}\right)\widetilde{e}_{p}\otimes a_{w}.
\end{align*}
Suppose the property true for an integer $n\geq1$, using the
isomorphism $A\otimes\bigotimes_{i=1}^{n+1}A_{i}\approx\left(A\otimes A_{1}\otimes\cdots\otimes A_{n}\right)\otimes A_{n+1}$
and the case $n=1$ we prove the result for $n+1$. 
\end{proof}
The previous proposition allows to represent algebraically the
inheritance of phenotypes which depend on several autosomal genes
and sex.
\begin{example}
\textit{Phenotypic difference according to sex}.

{\small A meta-analysis \cite{B-J-C-al} has shown that in the
human population, the autosomal genomes of men and women are
not significantly different, but in recent years it has become
clear that men and women are not equal when it comes to diseases.
Studies have shown that the incidence, severity or response to
treatment of cancers, cardiovascular, neurological or autoimmune
diseases are biased in favour of one sex or the other.}{\small\par}

We consider a phenotype in a bisexual population consisting of
diploid organisms. Let $\left\{ g_{1},\ldots,g_{m}\right\} $
be the set of autosomal genes controlling this phenotype and
for any $1\leq i\leq m$ let $g_{i}=\left\{ e_{i,1},\ldots,e_{i,k_{i}}\right\} $
be the set of alleles of the $g_{i}$ gene. The space $\text{span}\left(g_{i}\right)$
is equipped with the gametic algebra law $e_{i,p}e_{i,q}=\frac{1}{2}\left(e_{i,p}+e_{i,q}\right)$
from which the duplication $G_{i}=D\left(\text{span}\left(g_{i}\right)\right)$
define the zygotic algebra generated by $g_{i}$. For all $I,J\in\prod_{i=1}^{m}\left[\!\left[1,k_{i}\right]\!\right]$,
$I=\left(i_{1},\ldots,i_{m}\right)$, $J=\left(j_{1},\ldots,j_{m}\right)$
we note $e_{\left(I,J\right)}=\left(e_{1,i_{i}}\otimes e_{1,j_{i}}\right)\otimes\cdots\otimes\left(e_{m,i_{m}}\otimes e_{m,j_{m}}\right)$,
the family $\left(e_{\left(I,J\right)}\right)_{I,J}$ is therefore
a basis of genotype space $\bigotimes_{i=1}^{m}G_{i}$ .

Let $S$ be the gonosomic algebra defined on the basis $\left(f,m\right)$
by $f^{2}=m^{2}=0$ and $fm=mf=\frac{1}{2}f+\frac{1}{2}m$, then
$f\otimes e_{I,J}$ (resp. $m\otimes e_{i,J}$) represents a
female (resp. male) trait of the phenotype studied. We note $\pi\left(I,J\right)$
(resp. $\widetilde{\pi}\left(I,J\right)$) the prevalence, that
is to say the proportion of women (resp. men) presenting the
phenotype controlled by the genotype $e_{\left(I,I\right)}$.
\end{example}

Applying the proposition \ref{prop:Prod_tens} with $G=S\otimes\bigotimes_{i=1}^{m}G_{i}$
and $\Psi:G\rightarrow G$, $\Psi\left(f\otimes e_{\left(I,J\right)}\right)=\pi\left(I,J\right)f\otimes e_{\left(I,J\right)}$,
$\Psi\left(m\otimes e_{\left(I,J\right)}\right)=\widetilde{\pi}\left(I,J\right)m\otimes e_{\left(I,J\right)}$,
then for all $z,z'\in G$ the product $zz\lyxmathsym{‘}$ gives
the distribution of phenotypes in the offspring of a cross between
two individuals with phenotypes $z$ and $z\lyxmathsym{’}$.
\medskip{}

\section{duplication of gonosomic algebras.}

Let $A$ be a commutative $\mathbb{K}$-algebra, the \emph{non
commutative duplication} of $A$ is the space $A\otimes A$ and
the \emph{commutative duplication} of $A$ is the quotient space
of $A\otimes A$ by the ideal spanned by $\left\{ x\otimes y-y\otimes x;x,y\in A\right\} $.
Multiplication on non commutative and commutative duplication
is $\left(x\otimes y\right)\left(x'\otimes y'\right)=\left(xy\right)\otimes\left(x'y'\right)$.
The surjective morphism $\mu:A\otimes A\rightarrow A^{2}$, $x\otimes y\mapsto xy$
is called the \emph{Etherington morphism}. \medskip{}

The concept of duplication of a non associative algebra was introduced
by Etherington \cite{Ether-41a} to obtain the structure of genotypic
algebras from gametic algebras. Let us illustrate this with a
simple example. Let $e{{}_1}$ and $e{{}_2}$ be two alleles
of an autosomal gene in a population of diploid organisms. We
define the gametic algebra $A$ on the basis $\left(e_{1},e_{2}\right)$
by $e_{1}^{2}=e_{1}$, $e_{2}^{2}=e_{2}$ and $e_{1}e_{2}=\frac{1}{2}e_{1}+\frac{1}{2}e_{2}$.
The vectors $e_{i}\otimes e_{j},1\leq i\leq j\leq2$ of the basis
of the commutative duplication of $A$, represent the genotypes.
By definition of multiplication in the duplication, we have,
for example, $\left(e_{1}\otimes e_{2}\right)\left(e_{1}\otimes e_{2}\right)=\left(e_{1}e_{2}\right)\otimes\left(e_{1}e_{2}\right)=\frac{1}{4}e_{1}\otimes e_{1}+\frac{1}{2}e_{1}\otimes e_{2}+\frac{1}{4}e_{2}\otimes e_{2}$,
which corresponds to the distribution of genotypes given by Mendel's
second law.
\begin{rem}
The notion of duplication of an algebra allows us to interpret
the product of vectors from a genetic point of view. In genetic
algebra, it is customary to say that the product of two vectors
represents the crossing of two genetic types, so the product
$\left(e_{1}\otimes e_{2}\right)\left(e_{1}\otimes e_{2}\right)$
represents the crossbreeding of two organisms with genotype $e_{1}\otimes e_{2}$.
But using this interpretation, in gametic algebra, the product
$e_{i}e_{j}$ would represent the crossing of two gametes of
the same sex type, which biologically makes no sense. However,
with the Etherington morphism, from the relation $e_{i}e_{j}=\mu\left(e_{i}\otimes e_{j}\right)$
we deduce that the product $e_{i}e_{j}$ gives the distribution
of genetic types in the gametes produced by organisms of genotype
$e_{i}\otimes e_{j}$ and therefore the Etherington morphism
is interpreted as the gametogenesis operator.
\end{rem}

The definition of gonosomic algebra is not preserved by the duplication
of an algebra.
\begin{prop}
Let $A$ be a gonosomic $\mathbb{K}$-algebra, if $A^{2}\neq\left\{ 0\right\} $
then the duplication of $A$ is not gonosomic.
\end{prop}

\begin{proof}
Let $\left(e_{i}\right)_{i\in I}\cup\left(\widetilde{e}_{j}\right)_{j\in J}$
be a gonosomic basis of $A$. Suppose that the duplication $A\otimes A$
of $A$ is gonosomic, then for all $i\in I$ and $j\in J$ we
have $\left(e_{i}\otimes\widetilde{e}_{j}\right)^{2}=0$ thus
$\left(e_{i}\widetilde{e}_{j}\right)\otimes\left(e_{i}\widetilde{e}_{j}\right)=0$
which implies $e_{i}\widetilde{e}_{j}=0$ therefore $A^{2}=\left\{ 0\right\} $.
\end{proof}
In the following we denote by $D\left(A\right)$ the commutative
or not commutative duplication of an algebra $A$.
\begin{prop}
\label{prop:Duplicate}Duplication of a gonosomic algebra.

Let $A$ be a gonosomic $\mathbb{K}$-algebra and $\left(e_{i}\right)_{i\in I}\cup\left(\widetilde{e}_{j}\right)_{j\in J}$
a gonosomic basis of $A$. Let be $B=\text{span}\left(\left(e_{i}\right)_{i\in I}\right)$
and $\widetilde{B}=\text{span}\bigl(\left(\widetilde{e}_{j}\right)_{j\in J}\bigr)$.
If the vector space $B$ is equipped with an algebraic structure
by a law denoted $\bullet$, let $\mu,\widetilde{\mu}$ be the
morphisms defined by
\[
\mu:D\left(B\right)\rightarrow B^{2},\mu\left(x\otimes y\right)=x\bullet y\text{ and }\widetilde{\mu}:B\otimes\widetilde{B}\rightarrow A,\widetilde{\mu}\left(x\otimes y\right)=xy,
\]
given the linear maps $\varphi,\widetilde{\varphi}:B\rightarrow B$,
$\varphi':A\rightarrow B$ and $\widetilde{\varphi'}:A\rightarrow\widetilde{B}$.
Then the $\mathbb{K}$-vector space $D\left(B\right)\oplus B\otimes\widetilde{B}$
provided with the multiplication:
\begin{align*}
\left(x\oplus y\right)\left(x'\oplus y'\right)= & \left(\varphi\circ\mu\left(x\right)\otimes\varphi'\circ\widetilde{\mu}\left(y'\right)+\varphi\circ\mu\left(x'\right)\otimes\varphi'\circ\widetilde{\mu}\left(y\right)\right)\oplus\\
 & \quad\left(\widetilde{\varphi}\circ\mu\left(x\right)\otimes\widetilde{\varphi'}\circ\widetilde{\mu}\left(y'\right)+\widetilde{\varphi}\circ\mu\left(x'\right)\otimes\widetilde{\varphi'}\circ\widetilde{\mu}\left(y\right)\right)
\end{align*}
is a gonosomic algebra called the gonosomic duplication of the
gonosomic algebra $A$. 
\end{prop}

\begin{proof}
It is clear that for all $x,x'\in B$ we get $xx'=0$ and for
any $y,y'\in\widetilde{B}$ we get $yy'=0$. For all $i,j,k\in I$
and $p\in J$, we note $e_{\left(i,j\right)}=e_{i}\otimes e_{j}$
and $\widetilde{e}_{\left(i,p\right)}=e_{i}\otimes\widetilde{e}_{p}$,
given $\varphi\circ\mu\left(e_{\left(i,j\right)}\right)=\varphi\left(e_{i}\bullet e_{j}\right)=\sum_{u\in I}\lambda_{iju}e_{u}$,
$\widetilde{\varphi}\circ\mu\left(e_{\left(i,j\right)}\right)=\sum_{u\in I}\widetilde{\lambda}_{iju}e_{u}$
and $\varphi'\circ\widetilde{\mu}\left(\widetilde{e}_{\left(k,p\right)}\right)=\varphi'\left(e_{k}\widetilde{e}_{p}\right)=\sum_{v\in I}\xi_{kpv}e_{v}$,
$\widetilde{\varphi'}\circ\widetilde{\mu}\left(\widetilde{e}_{\left(k,p\right)}\right)=\sum_{q\in J}\widetilde{\xi}_{kpq}\widetilde{e}_{q}$
we have 
\begin{align*}
e_{\left(i,j\right)}\widetilde{e}_{\left(k,p\right)}=\left(e_{i}\otimes e_{j}\right)\left(e_{k}\otimes\widetilde{e}_{p}\right) & =\sum_{u,v\in I}\lambda_{iju}\xi_{kpv}e_{u}\otimes e_{v}+\sum_{\left(u,q\right)\in I\times J}\widetilde{\lambda}_{iju}\widetilde{\xi}_{kpq}e_{u}\otimes\widetilde{e}_{q}
\end{align*}
and noting $\gamma_{\left(i,j\right)\left(k,p\right)\left(u,v\right)}=\lambda_{iju}\xi_{kpv}$
and $\widetilde{\gamma}_{\left(i,j\right)\left(k,p\right)\left(u,q\right)}=\widetilde{\lambda}_{iju}\widetilde{\xi}_{kpq}$,
this can be written as 
\[
e_{\left(i,j\right)}\widetilde{e}_{\left(k,p\right)}=\sum_{u,v\in I}\gamma_{\left(i,j\right)\left(k,p\right)\left(u,v\right)}e_{\left(u,v\right)}+\sum_{\left(u,q\right)\in I\times J}\widetilde{\gamma}_{\left(i,j\right)\left(k,p\right)\left(u,q\right)}\widetilde{e}_{\left(u,q\right)}
\]
which proves that the space $D\left(B\right)\oplus B\otimes\widetilde{B}$
with this law is a gonosomic algebra. 
\end{proof}
The following examples illustrate the role of the maps $\mu,\widetilde{\mu},\varphi,\widetilde{\varphi},\varphi',\widetilde{\varphi}'$
in the duplication of a gonosomic algebra.
\begin{example}
\emph{Transmission of a X-linked multi allelic gene}.

This result is a good algebraic model of the reproduction of
diploid organisms in the XY-system. Consider a X-linked gene
with alleles $a_{1},\ldots,a_{n}$. Algebraically a maternal
genotype for this gene is $a_{i}\otimes a_{j}$ $\left(i\leq j\right)$
and a paternal genotype is $a_{p}\otimes Y$. We apply the proposition
\ref{prop:Duplicate} by taking $B=\text{span}\left(a_{1},\ldots,a_{n}\right)$
and $\widetilde{B}=\text{span}\left(Y\right)$, we define the
laws $a_{i}\bullet a_{j}=\frac{1}{2}a_{i}+\frac{1}{2}a_{j}$
and $a_{p}Y=\frac{1}{2}a_{p}+\frac{1}{2}Y$, which corresponds
to the meiosis results, the linear maps $\varphi=\widetilde{\varphi}=\text{id}_{B}$,
$\varphi',\widetilde{\varphi'}:B\oplus\widetilde{B}\rightarrow B\oplus\widetilde{B}$
such that $\varphi'\left(x\right)=\begin{cases}
x & \text{if }x\in B\\
0 & \text{if }x\in\widetilde{B}
\end{cases}$, $\widetilde{\varphi'}\left(x\right)=\begin{cases}
0 & \text{if }x\in B\\
x & \text{if }x\in\widetilde{B}
\end{cases}$ . Then on $B\otimes B\times B\otimes\widetilde{B}$ we get
\begin{align*}
\left(a_{i}\otimes a_{j},0\right)\left(0,a_{p}\otimes Y\right) & =\left(\left(\tfrac{1}{2}a_{i}+\tfrac{1}{2}a_{j}\right)\otimes\left(\tfrac{1}{2}a_{p}\right),\left(\tfrac{1}{2}a_{i}+\tfrac{1}{2}a_{j}\right)\otimes\left(\tfrac{1}{2}Y\right)\right)\\
 & =\left(\tfrac{1}{4}a_{i}\otimes a_{p}+\tfrac{1}{4}a_{j}\otimes a_{p},\tfrac{1}{4}a_{i}\otimes Y+\tfrac{1}{4}a_{j}\otimes Y\right)
\end{align*}
which gives the distribution of genotypes according to sex in
the offspring after crossing between a female of genotype $a_{i}a_{j}$
and a male $a_{p}Y$.
\end{example}

\medskip{}

\begin{example}
\emph{Transmission Ratio Distortion}

{\small Transmission ratio distortion (TRD) is a genetic phenomenon
induced by segregation distorters which are genes that bias Mendelian
segregation in their favour \cite{Dai-26}. In TRD one of the
two alleles from either parent is transmitted preferentially
to the offspring, this leads to a departure from the Mendelian
expected inheritance ratio 1/2. We distinguish two types of TRD
whether or not it depends on the sex of the parents.}{\small\par}

\medskip{}

Let $a{{}_1},\cdots,a_{n}$ be the alleles of an autosomal distorter
gene. For any $1\leq i\leq n$ we denote by $a_{i}$ (resp. $\widetilde{a}_{i}$)
the allele carried by a female (resp. male) gamete, let $A$
be the gonosomic algebra and $\left(a_{i}\right)_{i\in N_{n}}\cup\left(\widetilde{a}_{i}\right)_{i\in N_{n}}$
a gonosomic basis of $A$, $B=\text{span}\left(a{{}_1},\cdots,a_{n}\right)$
and $\widetilde{B}=\text{span}\left(\widetilde{a}{{}_1},\cdots,\widetilde{a}_{n}\right)$.
We take the commutative duplication $D\left(A\right)$ for the
space of female genotypes and $B\otimes\widetilde{B}$ for male
genotypes. For any $1\leq i\leq j\leq n$ we note $\nu_{i}$
and $\nu_{j}$ where $0\leq\nu_{i},\nu_{j}\leq\frac{1}{2}$,
the transmission rate of alleles $a_{i}$ and $a_{j}$ by a female
of genotype $a_{i}\otimes a_{j}$. And similarly for any $1\leq i,j\leq n$,
we note $\widetilde{\nu}_{i}$ and $\widetilde{\nu}_{j}$ the
transmission rate of alleles $a_{i}$ and $\widetilde{a}_{j}$
by a male of genotype $a_{i}\otimes\widetilde{a}_{j}$. To apply
the proposition \ref{prop:Duplicate}, we equip the space $B\oplus\widetilde{B}$
with the following gonosomic algebra structure: $a_{i}a_{j}=\widetilde{a}_{i}\widetilde{a}_{j}=0$
and $a_{i}\widetilde{a}_{j}=\widetilde{\nu}_{i}a_{i}+\widetilde{\nu}_{j}\widetilde{a}_{j}$.
Then we define on $B$ the multiplication $a_{i}\bullet a_{j}=\nu_{i}a_{i}+\nu_{j}a_{j}$,
we take $\varphi=\widetilde{\varphi}=\text{id}_{B}$ and $\varphi'\left(x\right)=\begin{cases}
\sigma x & \text{if }x\in B\\
0 & \text{if }x\in\widetilde{B}
\end{cases}$, $\widetilde{\varphi'}\left(x\right)=\begin{cases}
0 & \text{if }x\in B\\
\left(1-\sigma\right)x & \text{if }x\in\widetilde{B}
\end{cases}$ where $\sigma$ is the proportion of females in the population,
which can be expressed in terms of the sex ratio $\rho$ by $\sigma=\frac{1}{1+\rho}$.
With this we get
\begin{align*}
\left(a_{i}\otimes a_{j}\right)\left(a_{p}\otimes\widetilde{a}_{q}\right) & =\left(\nu_{i}a_{i}+\nu_{j}a_{j}\right)\otimes\sigma\widetilde{\nu}_{p}a_{p}+\left(\nu_{i}a_{i}+\nu_{j}a_{j}\right)\otimes\left(1-\sigma\right)\widetilde{\nu}_{q}\widetilde{a}_{q}\\
 & =\sigma\left(\nu_{i}\widetilde{\nu}_{p}a_{i}\otimes a_{p}+\nu_{j}\widetilde{\nu}_{p}a_{j}\otimes a_{p}\right)+\left(1-\sigma\right)\left(\nu_{i}\widetilde{\nu}_{q}a_{i}\otimes\widetilde{a}_{q}+\nu_{j}\widetilde{\nu}_{q}a_{j}\otimes\widetilde{a}_{q}\right)
\end{align*}
we get the distribution of genotypes according to sex in the
offspring after crossing between a female of genotype $a_{i}a_{j}$
and a male $a_{p}a_{q}$.
\end{example}

\medskip{}

\section{Gonosomic evolution operators.}

From now, unless otherwise stated, we assume that all gonosomic
$\mathbb{K}$-algebras considered are finite dimensional of type
$\left(n,m\right)$ with $\left(e_{i}\right)_{1\leq i\leq n}\cup\left(\widetilde{e}_{j}\right)_{1\leq j\leq m}$
as gonosomic basis. \medskip{}

To each gonosomic $\mathbb{K}$-algebra we can associate two
evolution operators. The first evolution operator $W$ is quadratic
and generates a discrete-time dynamical system linking the genetic
states of two successive generations. 

\medskip{}

Given a gonosomic $\mathbb{K}$-algebra $A$, we define the quadratic
operator $W$ called \textit{gonosomic evolution operator} by
\begin{equation}
\begin{array}{cccc}
W: & A & \rightarrow & A\\
 & z & \mapsto & \frac{1}{2}z^{2}.
\end{array}\label{eq:Op_W_def}
\end{equation}

\medskip{}

We call generation a biological cycle going from reproduction
to reproduction. In a bisexual panmictic population with discrete
non overlapping generations we consider a sex-linked gene whose
genetic female (resp. male) types are noted $\left(e_{i}\right)_{1\leq i\leq n}$
(resp. $\bigl(\widetilde{e}_{p}\bigr)_{1\leq p\leq m}$). For
a given $z\in A$ and for all $t\geq0$, $W^{t}\left(z\right)$
gives the genetic state of the population at the generation $t$.
The dynamical system generated by $W$ is defined by the sequence
$z$, $W\left(z\right)$, $W^{2}\left(z\right)$, $W^{3}\left(z\right)$
…\medskip{}

An element $z^{*}\in A$ is an \emph{equilibrium point} of the
dynamical system generated by $W$ if for all $t\geq1$ we have
$W^{t}\left(z^{*}\right)=z^{*}$.

It follows from the equivalence $W^{t}\left(z^{*}\right)=z^{*},\forall t\geq1\Leftrightarrow W\left(z^{*}\right)=z^{*}$
that $z^{*}$ is an equilibrium point if and only if $z^{*}$
is a fixed point of $W$.

From the definition of $W$ we deduce the following result. 
\begin{prop}
Let $A$ be gonosomic $\mathbb{K}$-algebra. There is one-to-one
correspondence between the idempotents of the algebra $A$ and
the fixed points of the gonosomic evolution operator $W$defined
on $A$. 
\end{prop}

\begin{proof}
If $e\in A$ is an idempotent, from the definition of $W$ we
get $W\left(2e\right)=2e$, i.e. $2e$ is a fixed point of $W$.
Conversely, if $z^{*}\in A$ is a fixed point of $W$, we have
$\left(\frac{1}{2}z^{*}\right)^{2}=\frac{1}{2}W\left(z^{*}\right)=\frac{1}{2}z^{*}$,
i.e. $\frac{1}{2}z^{*}$ is an idempotent of $A$. 
\end{proof}
Given $z\in A$, we note $z^{\left(0\right)}=z$ and $z^{\left(t\right)}=W^{t}\left(z\right)$
for all integer $t\geq0$, each $z^{\left(t\right)}$ corresponds
to a state of population at the generation $t$. We call trajectory
of the state $z^{\left(0\right)}$ for the gonosomic operator
$W$, the sequence $\bigl(z^{\left(t\right)}\bigr)_{t\in\mathbb{N}}$
and the challenge is to determine the limit points of the trajectories
$\left(W^{t}\left(z\right)\right)_{t\geq0}$ for any arbitrary
initial state $z\in A$. If the trajectory of the initial state
$z^{\left(0\right)}$ converge, there is a state noted $z^{\left(\infty\right)}$
such that $z^{\left(\infty\right)}=\lim_{t\rightarrow\infty}z^{\left(t\right)}$,
and by continuity of the operator $W$, the limit state $z^{\left(\infty\right)}$
is a fixed point of $W$.

In particular, if $\left(e_{i}\right)_{1\text{\ensuremath{\le}}i\text{\ensuremath{\le}}n}\cup\left(\widetilde{e}_{p}\right)_{1\text{\ensuremath{\le}}p\text{\ensuremath{\le}}m}$
is a gonosomic basis of $A$, for 
\[
z^{\left(t\right)}=W^{t}\left(z\right)=\sum_{i=1}^{n}x_{i}^{\:\left(t\right)}e_{i}+\sum_{p=1}^{m}y_{p}^{\:\left(t\right)}\widetilde{e}_{p}
\]
we find: 
\begin{eqnarray}
z^{\left(t+1\right)}=W\bigl(z^{\left(t\right)}\bigr) & = & \sum_{k=1}^{n}\sum_{i,p=1}^{n,m}\gamma_{ipk}x_{i}^{\:\left(t\right)}y_{p}^{\:\left(t\right)}e_{k}+\sum_{r=1}^{m}\sum_{i,p=1}^{n,m}\widetilde{\gamma}_{ipr}x_{i}^{\:\left(t\right)}y_{p}^{\:\left(t\right)}\widetilde{e}_{r}.\label{eq:W(x(t))}
\end{eqnarray}
The components of the operator $W$ correspond to the number
in the generation $t+1$ of females (resp. males) type $e_{k}$
(resp. $\widetilde{e}_{r}$) offsprings born after random mating
between all possible parents in generation $t$.

The quadratic evolution operator $W$ is defined in coordinate
form by: 
\[
\begin{array}{cccc}
W: & \mathbb{K}^{n}\times\mathbb{K}^{m} & \rightarrow & \mathbb{K}^{n}\times\mathbb{K}^{m}\\
 & \left(\left(x_{1},\ldots,x_{n}\right),\left(y_{1},\ldots,y_{m}\right)\right) & \mapsto & \left(\left(x_{1}',\ldots,x_{n}'\right),\left(y'_{1},\ldots,y'_{m}\right)\right)
\end{array}
\]
\begin{equation}
W:\left\{ \begin{aligned}x_{k}' & =\sum_{i,p=1}^{n,m}\gamma_{ipk}x_{i}y_{p},\quad k=1,\ldots,n\medskip\\
y'_{r} & =\sum_{i,p=1}^{n,m}\widetilde{\gamma}_{ipr}x_{i}y_{p},\quad r=1,\ldots,m,
\end{aligned}
\right.\label{eq:Op-W}
\end{equation}

Conversely, it is clear that any operator of the form (\ref{eq:Op-W})
is associated to a gonosomic algebra.

\medskip{}

Applying the linear form $\varpi$ defined in (\ref{eq:form_lin_pi_def-1})
to (\ref{eq:W(x(t))}) we find 
\begin{equation}
\varpi\bigl(z^{\left(t+1\right)}\bigr)=\varpi\circ W\left(z^{\left(t\right)}\right)=\sum_{i,p=1}^{n,m}\varpi\left(e_{i}\widetilde{e}_{p}\right)x_{i}^{\:\left(t\right)}y_{p}^{\:\left(t\right)}\label{eq:masse de x^(t)}
\end{equation}
which corresponds to the total population size at generation
$t+1$ depending on the genetic state $\left(\left(x_{i}^{\left(t\right)}\right)_{i\in N_{n}},\left(y_{j}^{\left(t\right)}\right)_{j\in N_{m}}\right)$
at generation $t$.

\medskip{}

For applications in genetics, we will assume in the following
that the field $\mathbb{K}$ is ordered, equipped with the order
topology and we denote by $\mathbb{K}_{+}$ the set $\left\{ x\in\mathbb{K};x\geq0\right\} $.
\begin{prop}
\label{prop:W(R+,R+)} Let $\mathbb{K}$ be an ordered field
and $A$ a gonosomic $\mathbb{K}$-algebra of type $\left(n,m\right)$,
we have $W\left(\mathbb{K}_{+}^{n}\times\mathbb{K}_{+}^{m}\right)\subset\mathbb{K}_{+}^{n}\times\mathbb{K}_{+}^{m}$
if and only if $\gamma_{ipk}\geq0$ and $\widetilde{\gamma}_{ipr}\geq0$
for all $i,k\in N_{n}$ and $p,r\in N_{m}$.
\end{prop}

\begin{proof}
The necessary condition follows from $W\left(e_{i}+\widetilde{e}_{p}\right)=\sum_{k}\gamma_{ipk}e_{k}+\sum_{r}\widetilde{\gamma}_{ipr}\widetilde{e}_{r}$
for all $i\in N_{n}$ and $p\in N_{\nu}$. The sufficient condition
immediately follows from (\ref{eq:Op-W}).
\end{proof}
This result leads to the following definition:
\begin{defn}
\label{def:non_negative}Let $\mathbb{K}$ be an ordered field,
we say that a gonosomic $\mathbb{K}$-algebra $A$ of type $\left(n,m\right)$
is \textit{non negative} if it satisfies the definition \ref{def:Gonosomic-Alg}
with $\gamma_{ipk}\geq0$, $\widetilde{\gamma}_{ipr}\geq0$ for
all $i,k\in N_{n}$ and $p,r\in N_{m}$. 
\end{defn}

From now the gonosomic algebras considered are non negative equipped
with a gonosomic basis $\left(e_{i}\right)_{1\text{\ensuremath{\le}}i\text{\ensuremath{\le}}n}\cup\left(\widetilde{e}_{j}\right)_{1\text{\ensuremath{\le}}j\text{\ensuremath{\le}}m}$
. 
\begin{example}
\label{exa:W-male_infertility}Let $A$ be the gonosomic $\mathbb{K}$-algebra
defined in example \ref{exa:male_infertility-1}, the gonosomic
evolution operator $W$ associated with this algebra is $\left(x'_{1},\left(y'_{1},y'_{2}\right)\right)=W\left(x_{1},\left(y_{1},y_{2}\right)\right)$
where $\left(x_{1},\left(y_{1},y_{2}\right)\right)\in\mathbb{K}_{+}\times\mathbb{K}_{+}^{2}$
and
\[
\left\{ \begin{aligned}x'_{1} & =\tfrac{1}{2}x_{1}y_{1},\\
y'_{1} & =\tfrac{1-\mu}{2}x_{1}y_{1},\\
y'_{2} & =\tfrac{\mu}{2}x_{1}y_{1}.
\end{aligned}
\right.
\]

We show by induction that for all $t\in\mathbb{N}$,
\[
W^{t+1}\left(x_{1},\left(y_{1},y_{2}\right)\right)=\left(\tfrac{2}{1-\mu}\left(\tfrac{1-\mu}{4}x_{1}y_{1}\right)^{2^{t}},2\left(\tfrac{1-\mu}{4}x_{1}y_{1}\right)^{2^{t}},\tfrac{2\mu}{1-\mu}\left(\tfrac{1-\mu}{4}x_{1}y_{1}\right)^{2^{t}}\right)
\]
it follows that 
\[
\lim_{t\rightharpoondown+\infty}W^{t}\left(x_{1},\left(y_{1},y_{2}\right)\right)=\begin{cases}
\left(0,\left(0,0\right)\right) & \text{if }\left(1-\mu\right)x_{1}y_{1}<4,\\
\left(\frac{2}{1-\mu},\left(2,\frac{2\mu}{1-\mu}\right)\right) & \text{if }\left(1-\mu\right)x_{1}y_{1}=4,\\
\left(+\infty,\left(+\infty,+\infty\right)\right) & \text{if }\left(1-\mu\right)x_{1}y_{1}>4.
\end{cases}
\]

It follows that male infertility will disappear only if the population
disappears.
\end{example}

\medskip{}

\begin{example}
\label{exa:W-bidirectional}For example \ref{exa:bidirectional},
the gonosomic evolution operator $W$ defined on $\mathbb{K}_{+}^{2}\times\mathbb{K}_{+}^{2}$
is 
\[
\left\{ \begin{aligned}x'_{1}=y'_{1} & =\tfrac{1}{2}x_{1}y_{1},\\
x'_{2}=y'_{2} & =\tfrac{1}{2}x_{2}y_{2}.
\end{aligned}
\right.
\]

It is easy to see that $x_{1}^{\left(t\right)}=y_{1}^{\left(t\right)}=2\left(\frac{1}{4}x_{1}y_{1}\right)^{2^{t-1}}$
and $x_{2}^{\left(t\right)}=y_{2}^{\left(t\right)}=2\left(\frac{1}{4}x_{2}y_{2}\right)^{2^{t-1}}$
it follows that

\[
\lim_{t\rightharpoondown+\infty}x_{1}^{\left(t\right)}=\begin{cases}
0 & \text{if }x_{1}y_{1}<4,\\
2 & \text{if }x_{1}y_{1}=4,\\
+\infty & \text{if }x_{1}y_{1}>4,
\end{cases}\text{ and }\lim_{t\rightharpoondown+\infty}x_{2}^{\left(t\right)}=\begin{cases}
0 & \text{if }x_{2}y_{2}<4,\\
2 & \text{if }x_{2}y_{2}=4,\\
+\infty & \text{if }x_{2}y_{2}>4.
\end{cases}
\]
then the table of limits of the operator $W$ for $z=(x_{1},x_{2}),(y_{1},y_{2})$
is \medskip{}

$\underset{t\rightharpoondown+\infty}{\lim}W^{t}\left(z\right)=$%
\begin{tabular}{ccc|c}
$\text{ if }x_{1}y_{1}<4$ & $\text{ if }x_{1}y_{1}=4$ & \multicolumn{1}{c}{$\text{ if }x_{1}y_{1}>4$} & \tabularnewline
\cline{1-3}
$\left(\left(0,0\right),\left(0,0\right)\right)$ & $\left(\left(2,0\right),\left(2,0\right)\right)$ & $\left(\left(\infty,0\right),\left(\infty,0\right)\right)$ & $\text{ and }x_{2}y_{2}<4$\tabularnewline
$\left(\left(0,2\right),\left(0,2\right)\right)$ & $\left(\left(2,2\right),\left(2,2\right)\right)$ & $\left(\left(\infty,2\right),\left(\infty,2\right)\right)$ & $\text{ and }x_{2}y_{2}=4$\tabularnewline
$\left(\left(0,\infty\right),\left(0,\infty\right)\right)$ & $\left(\left(2,\infty\right),\left(2,\infty\right)\right)$ & $\left(\left(\infty,\infty\right),\left(\infty,\infty\right)\right)$ & $\text{ and }x_{2}y_{2}>4$\tabularnewline
\end{tabular}

We see that, except in the case where the population disappears,
in all other cases at least one of the two alleles is always
maintained in both sexes.
\end{example}

It is not always possible, as in the above example, to provide
an explicit expression for the trajectories of a gonosomic operator;
in such cases, suitable bounding arguments may be used. 
\begin{example}
\label{exa:W-Hybrid-dysgenesis}The gonosomic evolution operator
associated with the gonosomic algebra defined in example \ref{exa:Hybrid-dysgenesis}
is 
\[
\left\{ \begin{aligned}x'_{1}=y'_{1} & =\tfrac{1}{2}x_{1}y_{1},\\
x'_{2}=y'_{2} & =\tfrac{1}{2}\text{\ensuremath{\left(y_{1}+y_{2}\right)x_{2}}},
\end{aligned}
\right.
\]

Let be $\bigl(\bigl(x_{1}^{\left(t\right)},x_{2}^{\left(t\right)}\bigr),\bigl(y_{1}^{\left(t\right)},y_{2}^{\left(t\right)}\bigr)\bigr)=W^{t}\left(\left(x_{1},x_{2}\right),\left(y_{1},y_{2}\right)\right)$,
for all $t\geq1$ we have $x_{1}^{\left(t\right)}=y_{1}^{\left(t\right)}$
et $x_{2}^{\left(t\right)}=y_{2}^{\left(t\right)}$. From $x'_{1}=\tfrac{1}{2}x_{1}y_{1}$
it follows that 
\[
x_{1}^{\left(t\right)}=\left(\tfrac{1}{2}\right)^{2^{t}-1}\left(x_{1}y_{1}\right)^{2^{t-1}}=2\left(\tfrac{1}{4}x_{1}y_{1}\right)^{2^{t-1}}.
\]

Therefore,
\[
\lim_{t\rightharpoondown+\infty}x_{1}^{\left(t\right)}=\lim_{t\rightharpoondown+\infty}y_{1}^{\left(t\right)}=\begin{cases}
0 & \text{if }x_{1}y_{1}<4,\\
2 & \text{if }x_{1}y_{1}=4,\\
+\infty & \text{if }x_{1}y_{1}>4.
\end{cases}
\]

If $x_{2}=0$ we get $x_{2}^{\left(t\right)}=0$ for any $t\geq1$.
Now we assume that $x_{2}\neq0$.

If $\underset{t\rightharpoondown+\infty}{\lim}x_{1}^{\left(t\right)}=0$
we have $\frac{1}{4}x_{1}y_{1}<1$. For any $x$ such that $0\leq x\leq2-\frac{x_{1}y_{1}}{4}$
we have $\frac{1}{2}\left(\frac{x_{1}y_{1}}{4}+x\right)x\leq x$.
Therefore, if $x_{2}\leq2-\frac{x_{1}y_{1}}{4}$ we have $\frac{1}{2}\left(\frac{x_{1}y_{1}}{4}+x_{2}\right)x_{2}\leq x_{2}$
in other words $x_{2}^{\left(1\right)}\leq x_{2}$, from this
it follows that $x_{2}^{\left(1\right)}\leq2-\frac{x_{1}y_{1}}{4}$
so $x_{2}^{\left(2\right)}\leq x_{2}^{\left(1\right)}$ and recursively
the sequence $\left(x_{2}^{\left(t\right)}\right)_{t}$ is decreasing,
therefore we get $\underset{t\rightharpoondown+\infty}{\lim}x_{2}^{\left(t\right)}=0$.
If $x_{2}>2-\frac{x_{1}y_{1}}{4}$ we have $x_{2}^{\left(1\right)}>x_{2}$
and we deduce recursively that the sequence $\left(x_{2}^{\left(t\right)}\right)_{t}$
is increasing with $\underset{t\rightharpoondown+\infty}{\lim}x_{2}^{\left(t\right)}=+\infty$.

If $\underset{t\rightharpoondown+\infty}{\lim}x_{1}^{\left(t\right)}=2$
we have $x_{1}^{\left(t\right)}=2$ for all $t\geq1$, so $x_{2}^{\left(t+1\right)}=\frac{1}{2}\left(x_{2}^{\left(t\right)}+1\right)^{2}-\frac{1}{2}$.
Setting $w_{t}=x_{2}^{\left(t\right)}+1$ we have $w_{t+1}=w_{t}^{2}+\frac{1}{2}$
from which we conclude that $\underset{t\rightharpoondown+\infty}{\lim}w_{t}=+\infty$
therefore $\underset{t\rightharpoondown+\infty}{\lim}x_{2}^{\left(t\right)}=+\infty$. 

If $\underset{t\rightharpoondown+\infty}{\lim}x_{1}^{\left(t\right)}=+\infty$,
from $x_{2}^{\left(t+1\right)}=\tfrac{1}{2}\text{\ensuremath{\left(x_{1}^{\left(t\right)}+x_{2}^{\left(t\right)}\right)x_{2}^{\left(t\right)}}\ensuremath{\ensuremath{\geq\frac{1}{2}x_{1}^{\left(t\right)}x_{2}^{\left(t\right)}\geq}}}\left(\tfrac{1}{4}x_{1}y_{1}\right)x_{2}^{\left(t\right)}$
it comes $x_{2}^{\left(t+1\right)}\geq\left(\tfrac{1}{4}x_{1}y_{1}\right)^{t}x_{2}$
therefore$\underset{t\rightharpoondown+\infty}{\lim}x_{2}^{\left(t\right)}=+\infty$.

In conclusion, we obtained
\[
\underset{t\rightharpoondown+\infty}{\lim}x_{2}^{\left(t\right)}=\begin{cases}
0 & \text{if }x_{1}y_{1}<4\text{ and }x_{2}\leq2-\frac{x_{1}y_{1}}{4},\\
+\infty & \text{if }x_{1}y_{1}<4\text{ and }x_{2}>2-\frac{x_{1}y_{1}}{4},\\
+\infty & \text{if }x_{1}y_{1}\geq4.
\end{cases}
\]

According to this result, the persistence of gene $P$ in the
population depends on a critical population size of the group
carrying the gene $M$.
\end{example}

\begin{prop}
\label{prop:x_k and y_k}Let $A$ be a non negative gonosomic
algebra and $\left(e_{i}\right)_{1\text{\ensuremath{\le}}i\text{\ensuremath{\le}}n}\cup\left(\widetilde{e}_{j}\right)_{1\text{\ensuremath{\le}}j\text{\ensuremath{\le}}m}$
a gonosomic basis of $A$. For all $z\in\mathbb{K}_{+}^{n}\times\mathbb{K}_{+}^{m}$,
$z=\left(\left(x_{1},\ldots,x_{n}\right),\left(y_{1},\ldots,y_{m}\right)\right)$
and $t\geq1$, if we denote $W^{t}\bigl(z\bigr)=\left(x_{1}^{\:\left(t\right)},\ldots,x_{n}^{\:\left(t\right)},y_{1}^{\;\left(t\right)},\ldots,y_{m}^{\;\left(t\right)}\right)$,
then we have 
\[
0<\underset{\left(i,p\right)\in E_{k}}{\min}\left\{ \gamma_{ipk}\right\} \Bigl(\sum_{i=1}^{n}x_{i}^{\left(t-1\right)}\Bigr)\Bigl(\sum_{p=1}^{m}y_{p}^{\left(t-1\right)}\Bigr)\leq x_{k}^{\left(t\right)}\leq\underset{i,p}{\max}\left\{ \gamma_{ipk}\right\} \Bigl(\sum_{i=1}^{n}x_{i}^{\left(t-1\right)}\Bigr)\Bigl(\sum_{p=1}^{m}y_{p}^{\left(t-1\right)}\Bigr)
\]
\[
0<\underset{\left(i,p\right)\in\widetilde{E}_{r}}{\min}\left\{ \widetilde{\gamma}_{ipr}\right\} \Bigl(\sum_{i=1}^{n}x_{i}^{\left(t-1\right)}\Bigr)\Bigl(\sum_{p=1}^{m}y_{p}^{\left(t-1\right)}\Bigr)\leq y_{r}^{\left(t\right)}\leq\underset{i,p}{\max}\left\{ \widetilde{\gamma}_{ipr}\right\} \Bigl(\sum_{i=1}^{n}x_{i}^{\left(t-1\right)}\Bigr)\Bigl(\sum_{p=1}^{m}y_{p}^{\left(t-1\right)}\Bigr),
\]
where $E_{k}=\left\{ \left(i,p\right)\in N_{n}\times N_{m};\gamma_{ipk}>0\right\} $
and $\widetilde{E}_{r}=\left\{ \left(i,p\right)\in N_{n}\times N_{m};\widetilde{\gamma}_{ipr}>0\right\} $.
\end{prop}

\begin{proof}
Let for any $t\geq1$ be $W^{t-1}\left(z\right)=\sum_{i=1}^{n}x_{i}^{\left(t-1\right)}e_{i}+\sum_{p=1}^{m}y_{p}^{\left(t-1\right)}\widetilde{e}_{p}$.
With the proposition (\ref{prop:W(R+,R+)}) we prove by induction
that $x_{i}^{\left(t\right)},y_{p}^{\left(t\right)}\geq0$ for
all $t\geq0$. From (\ref{eq:Op-W}) we get 
\begin{align*}
x_{k}^{\left(t\right)} & =\sum_{i,p=1}^{n,m}\gamma_{ipk}x_{i}^{\left(t-1\right)}y_{p}^{\left(t-1\right)},\enspace\left(k\in N_{n}\right)
\end{align*}

From this we deduce for each $k\in N_{n}$ that 
\[
0<\underset{\left(i,p\right)\in E_{k}}{\min}\left\{ \gamma_{ipk}\right\} \Bigl(\sum_{i,p=1}^{n,m}x_{i}^{\left(t-1\right)}y_{p}^{\left(t-1\right)}\Bigr)\leq x_{k}^{\left(t\right)}\leq\underset{i,p}{\max}\left\{ \gamma_{ipk}\right\} \Bigl(\sum_{i,p=1}^{n,m}x_{i}^{\left(t-1\right)}y_{p}^{\left(t-1\right)}\Bigr).
\]

A similar reasoning gives the inequalities for $y_{r}^{\;\left(t\right)}$. 
\end{proof}
\medskip{}

Using the linear form $\varpi:A\rightarrow\mathbb{K}$ defined
in (\ref{eq:form_lin_pi_def-1}), for a given non negative gonosomic
algebra with a gonosomic basis $\left(e_{i}\right)_{1\text{\ensuremath{\le}}i\text{\ensuremath{\le}}n}\cup\left(\widetilde{e}_{j}\right)_{1\text{\ensuremath{\le}}j\text{\ensuremath{\le}}m}$
and $e_{i}\widetilde{e}_{j}=\sum_{k=1}^{n}\gamma_{ijk}e_{k}+\sum_{p=1}^{m}\widetilde{\gamma}_{ijp}\widetilde{e}_{p}$
for all $i\in N_{n}$ and $j\in N_{m}$, the scalar
\[
\varpi\left(e_{i}\widetilde{e}_{j}\right)=\sum_{k=1}^{n}\gamma_{ijk}+\sum_{p=1}^{m}\widetilde{\gamma}_{ijp}
\]
 gives the total frequency of offspring resulting from the mating
of a female $e_{i}$ with a male $\widetilde{e}_{j}$. For all
$i\in N_{n}$ and $j\in N_{m}$ we note
\[
\gamma_{ij}=\sum_{k=1}^{n}\gamma_{ijk}\;\left(\text{resp. }\widetilde{\gamma}_{ij}=\sum_{p=1}^{m}\widetilde{\gamma}_{ijp}\right)
\]
the total frequency of females (resp. males) resulting from the
mating between a female $e_{i}$ and a male $\widetilde{e}_{j}$.
Given an initial population size $z\in\mathbb{K}_{+}^{n}\times\mathbb{K}_{+}^{m}$,
the real $\varpi\circ W^{t}\left(z\right)$ gives the total population
size at generation $t$. 
\begin{prop}
\label{prop:max-woW} Let $\mathbb{K}$ be an ordered field and
$A$ a non negative gonosomic $\mathbb{K}$-algebra of type $\left(n,m\right)$.
For $z\in\mathbb{K}_{+}^{n}\times\mathbb{K}_{+}^{m}$, $z=\left(\left(x_{1},\ldots,x_{n}\right),\left(y_{1},\ldots,y_{m}\right)\right)$
and for all $t\geq1$ we have 
\begin{align*}
i)\quad\varpi\circ W^{t}\left(z\right) & \leq\left(\dfrac{1}{4}\underset{i,p}{\max}\left\{ \varpi\left(e_{i}\widetilde{e}_{p}\right)\right\} \right)^{2^{t}-1}\varpi\left(z\right)^{2^{t}},\\
{\small ii)\quad\varpi\circ W^{t}\bigl(z\bigr)} & \leq\left(\underset{i,j}{\max}\left\{ \varpi\left(e_{i}\widetilde{e}_{j}\right)\right\} \right)^{\frac{1}{3}\left(4^{\left\lfloor \nicefrac{\left(t+1\right)}{2}\right\rfloor }-1\right)}\left(\frac{1}{16}\max_{i,j,p,q}\left\{ \gamma_{ij}\widetilde{\gamma}_{pq}\right\} \right)^{\frac{1}{3}\left(4^{\left\lfloor \nicefrac{t}{2}\right\rfloor }-1\right)}\times\\
 & {\small }\hspace{65mm}\times\begin{cases}
\;\;\Bigl(\varpi\bigl(z\bigr)\Bigr){}^{4^{\left\lfloor \nicefrac{t}{2}\right\rfloor }} & \mbox{if }t\mbox{ is even,}\medskip\\
\Bigl(\tfrac{1}{4}\varpi\bigl(z\bigr)\Bigr){}^{4^{\left\lfloor \nicefrac{t}{2}\right\rfloor }} & \mbox{if }t\mbox{ is odd}.
\end{cases}
\end{align*}
\end{prop}

\begin{proof}
Let $\left(e_{i}\right)_{1\text{\ensuremath{\le}}i\text{\ensuremath{\le}}n}\cup\left(\widetilde{e}_{j}\right)_{1\text{\ensuremath{\le}}j\text{\ensuremath{\le}}m}$
be a gonosomic basis of $A$. For $z\in A$, $z=\sum_{i=1}^{n}x_{i}e_{i}+\sum_{j=1}^{m}y_{j}\widetilde{e}_{j}$,
using (\ref{eq:Op-W}) we have 
\[
W\left(z\right)=\sum_{k=1}^{n}\sum_{i,j=1}^{n,m}\gamma_{ijk}x_{i}y_{j}e_{k}+\sum_{r=1}^{m}\sum_{i,j=1}^{n,m}\gamma_{ijr}x_{i}y_{j}\widetilde{e}_{r}=\sum_{k=1}^{n}x'_{k}e_{k}+\sum_{r=1}^{m}y'_{r}\widetilde{e}_{r}.
\]

From this it follows 
\[
\varpi\circ W\left(z\right)=\sum_{k,r=1}^{n,m}\varpi\left(e_{k}\widetilde{e}_{r}\right)x{}_{k}y{}_{r}
\]

and 
\[
\varpi\circ W^{2}\left(z\right)=\sum_{k,r=1}^{n,m}\varpi\left(e_{k}\widetilde{e}_{r}\right)x'_{k}y'_{r}.
\]

We have $\varpi\left(e_{k}\widetilde{e}_{r}\right)\geq0$, $x'_{k}\geq0$
and $y'_{r}\geq0$ for all $k\in N_{n}$ and $r\in N_{m}$, thus
\[
\underset{k,r}{\min}\left\{ \varpi\left(e_{k}\widetilde{e}_{r}\right)\right\} \sum_{k,r=1}^{n,m}x{}_{k}y{}_{r}\leq\varpi\circ W\left(z\right)\leq\underset{k,r}{\max}\left\{ \varpi\left(e_{k}\widetilde{e}_{r}\right)\right\} \sum_{k,r=1}^{n,m}x{}_{k}y{}_{r}
\]
\begin{equation}
\underset{k,r}{\min}\left\{ \varpi\left(e_{k}\widetilde{e}_{r}\right)\right\} \sum_{k,r=1}^{n,m}x'_{k}y'_{r}\leq\varpi\circ W^{2}\left(z\right)\leq\underset{k,r}{\max}\left\{ \varpi\left(e_{k}\widetilde{e}_{r}\right)\right\} \sum_{k,r=1}^{n,m}x'_{k}y'_{r}.\label{eq:woW^2}
\end{equation}

First we have 
\[
\sum_{k,r=1}^{n,m}x{}_{k}y{}_{r}=\biggl(\sum_{k=1}^{n}x{}_{k}\biggr)\biggl(\sum_{r=1}^{m}y{}_{r}\biggr)\leq\frac{1}{4}\biggl(\sum_{k=1}^{n}x{}_{k}+\sum_{r=1}^{m}y{}_{r}\biggr)^{2}
\]

thus we get 
\begin{equation}
\varpi\circ W\left(z\right)\leq\underset{k,r}{\max}\left\{ \varpi\left(e_{k}\widetilde{e}_{r}\right)\right\} \left(\frac{1}{2}\varpi\left(z\right)\right)^{2},\label{eq:woW(z)-maj}
\end{equation}

from this we deduce recursively the inequality $i)$. Next with
(\ref{eq:Op-W}) we get 
\begin{align}
\sum_{k,r=1}^{n,m}x'_{k}y'_{r} & =\sum_{k,r=1}^{n,m}\biggl(\sum_{i,j=1}^{n,m}\gamma_{ijk}x_{i}y_{j}\biggr)\biggl(\sum_{p,q=1}^{n,m}\gamma_{pqr}x_{p}y_{q}\biggr)=\sum_{i,p=1}^{n}\sum_{j,q=1}^{m}\gamma_{ij}\widetilde{\gamma}_{pq}\:x_{i}x_{p}y_{j}y_{q}.\label{eq:xixpyjyq}
\end{align}

We have $\gamma_{ij}\geq0$ and $\widetilde{\gamma}_{pq}\geq0$
for every $i,p\in N_{n}$ and $j,q\in N_{m}$, thus 
\[
\sum_{i,p=1}^{n}\sum_{j,q=1}^{m}\gamma_{ij}\widetilde{\gamma}_{pq}\:x_{i}x_{p}y_{j}y_{q}\leq\max_{i,j,p,q}\left\{ \gamma_{ij}\widetilde{\gamma}_{pq}\right\} \Bigl(\sum_{i=1}^{n}x_{i}\Bigr)^{2}\Bigl(\sum_{j=1}^{m}y_{j}\Bigr)^{2}
\]
using $ab\leq\frac{1}{4}\left(a+b\right)^{2}$ we get $\Bigl(\sum_{i=1}^{n}x_{i}\Bigr)^{2}\Bigl(\sum_{j=1}^{m}y_{j}\Bigr)^{2}\leq\frac{1}{16}\left(\sum_{i=1}^{n}x_{i}+\sum_{j=1}^{m}y_{j}\right)^{4}$
where $\sum_{i=1}^{n}x_{i}+\sum_{j=1}^{m}y_{j}=\varpi\left(z\right)$,
finally 
\[
\varpi\circ W^{2}\left(z\right)\leq\underset{k,r}{\max}\left\{ \varpi\left(e_{k}\widetilde{e}_{r}\right)\right\} \times\max_{i,j,p,q}\left\{ \gamma_{ij}\widetilde{\gamma}_{pq}\right\} \left(\frac{1}{2}\varpi\left(z\right)\right)^{4}.
\]

It follows from this that for all integer $t\geq2$ we have 
\begin{equation}
\varpi\circ W^{t}\left(z\right)\leq\left(\tfrac{1}{2}\right)^{4}\underset{k,r}{\max}\left\{ \varpi\left(e_{k}\widetilde{e}_{r}\right)\right\} \times\max_{i,j,p,q}\left\{ \gamma_{ij}\widetilde{\gamma}_{pq}\right\} \left(\varpi\circ W^{t-2}\left(z\right)\right)^{4}\label{eq:woW^t}
\end{equation}

With (\ref{eq:woW(z)-maj}) and (\ref{eq:woW^t}) we establish
by induction the inequality given in the proposition.
\end{proof}
We introduce the sets 
\begin{align}
\mathscr{M} & =\left\{ \left(i,j\right)\in N_{n}\times N_{m};\varpi\left(e_{i}\widetilde{e}_{j}\right)=0\right\} ,\label{eq:MetN}\\
\mathscr{N} & =\left\{ \left(i,j\right)\in N_{n}\times N_{m};\varpi\left(e_{i}\widetilde{e}_{j}\right)\neq0\right\} .\nonumber 
\end{align}

In a non negative gonosomic $\mathbb{K}$-algebra with gonosomic
basis $\left(e_{i}\right)_{1\text{\ensuremath{\le}}i\text{\ensuremath{\le}}n}\cup\left(\widetilde{e}_{j}\right)_{1\text{\ensuremath{\le}}j\text{\ensuremath{\le}}m}$,
for any $\left(i,j\right)\in\mathscr{M}$ we have $e_{i}\widetilde{e}_{j}=0$,
genetically this means that the crossing between a female of
type $e_{i}$ and a male $\widetilde{e}_{p}$ is sterile. Therefore
the set $\mathscr{M}$ lists the sterile crosses. 
\begin{prop}
\label{prop:min-woW}Let $\mathbb{K}$ be a formally real field
and $A$ be a non negative gonosomic algebra of type $\left(n,m\right)$.
For all $z\in\mathbb{K}_{+}^{n}\times\mathbb{K}_{+}^{m}$, $z=\left(\left(x_{1},\ldots,x_{n}\right),\left(y_{1},\ldots,y_{m}\right)\right)$
and $t\geq1$ we have 
\[
\underset{\left(i,j\right)\in\mathscr{N}}{\min}\left\{ \varpi\left(e_{i}\widetilde{e}_{j}\right)\right\} \left(\underset{\left(i,j\right)\in\mathscr{N}}{\min}\left\{ \sqrt{\gamma_{ij}\widetilde{\gamma}_{ij}}\right\} \right)^{2^{t}-2}\biggl(\sum_{i,j=1}^{n,m}x_{i}y_{j}\biggr)^{2^{t-1}}\leq\varpi\circ W^{t}\left(z\right).
\]
\end{prop}

\begin{proof}
After exchanging roles between the couples $\left(i,j\right)$
and $\left(p,q\right)$ in (\ref{eq:xixpyjyq}) we get 
\[
\sum_{k,r=1}^{n,m}x'_{k}y'_{r}=\sum_{i,p=1}^{n}\sum_{j,q=1}^{m}\gamma_{pq}\widetilde{\gamma}_{ij}\:x_{i}x_{p}y_{j}y_{q}
\]

hence 
\[
\sum_{k,r=1}^{n,m}x'_{k}y'_{r}=\sum_{i,p=1}^{n}\sum_{j,q=1}^{m}\tfrac{1}{2}\left(\gamma_{ij}\widetilde{\gamma}_{pq}+\gamma_{pq}\widetilde{\gamma}_{ij}\right)\:x_{i}x_{p}y_{j}y_{q}
\]

using the relation $a+b\geq2\sqrt{ab}$ we get 
\begin{align*}
\sum_{k,r=1}^{n,m}x'_{k}y'_{r} & \geq\sum_{i,p=1}^{n}\sum_{j,q=1}^{m}\sqrt{\gamma_{ij}\gamma_{pq}\widetilde{\gamma}_{ij}\widetilde{\gamma}_{pq}}\:x_{i}x_{p}y_{j}y_{q}\geq\left(\underset{\left(i,j\right)\in\mathscr{N}}{\min}\left\{ \sqrt{\gamma_{ij}\widetilde{\gamma}_{ij}}\right\} \right)^{2}\biggl(\sum_{i,j=1}^{n,m}x_{i}y_{j}\biggr)^{2}
\end{align*}

It follows that for all integer $t\geq1$ 
\[
\left(\underset{\left(i,j\right)\in\mathscr{N}}{\min}\left\{ \sqrt{\gamma_{ij}\widetilde{\gamma}_{ij}}\right\} \right)^{2}\biggl(\sum_{i,j=1}^{n,m}x_{i}^{\left(t-1\right)}y_{j}^{\left(t-1\right)}\biggr)^{2}\leq\sum_{k,r=1}^{n,m}x{}_{k}^{\left(t\right)}y{}_{r}^{\left(t\right)}
\]

and by induction 
\[
\left(\underset{\left(i,j\right)\in\mathscr{N}}{\min}\left\{ \sqrt{\gamma_{ij}\widetilde{\gamma}_{ij}}\right\} \right)^{2^{t}}\biggl(\sum_{i,j=1}^{n,m}x_{i}y_{j}\biggr)^{2^{t}}\leq\sum_{k,r=1}^{n,m}x{}_{k}^{\left(t\right)}y{}_{r}^{\left(t\right)}.
\]

But according to (\ref{eq:woW^2}) we have for all $t\geq2$
\[
\underset{\left(i,j\right)\in\mathscr{N}}{\min}\left\{ \varpi\left(e_{i}\widetilde{e}_{j}\right)\right\} \left(\underset{\left(i,j\right)\in\mathscr{N}}{\min}\left\{ \sqrt{\gamma_{ij}\widetilde{\gamma}_{ij}}\right\} \right)^{2}\biggl(\sum_{i,j=1}^{n,m}x_{i}^{\left(t-2\right)}y_{j}^{\left(t-2\right)}\biggr)^{2}\leq\varpi\circ W^{t}\left(z\right).
\]
With these last two relations we get by induction the inequality
$ii)$ . 
\end{proof}
\begin{cor}
\label{cor:suite-woW^t}Let $\mathbb{K}$ be a formally real
field and $A$ a non negative gonosomic $\mathbb{K}$-algebra
of type $\left(n,m\right)$. For $z\in\mathbb{K}_{+}^{n}\times\mathbb{K}_{+}^{m}$
, $z=\left(\left(x_{1},\ldots,x_{n}\right),\left(y_{1},\ldots,y_{m}\right)\right)$.

a) If $\varpi\left(z\right)\leq\dfrac{4}{\underset{i,j}{\max}\left\{ \varpi\left(e_{i}\widetilde{e}_{j}\right)\right\} }$
then the sequence $\left(\varpi\circ W^{t}\left(z\right)\right)_{t\geq0}$
is decreasing.

b) If $\Bigl(\underset{\left(i,j\right)\in\mathscr{N}}{\min}\left\{ \sqrt{\gamma_{ij}\widetilde{\gamma}_{ij}}\right\} \Bigr)^{2}\Bigl(\sum_{i,j=1}^{n,m}x_{i}y_{j}\Bigr)>1$
then the sequence $\left(\varpi\circ W^{t}\left(z\right)\right)_{t\geq0}$
is divergent.

c) If $\sqrt[3]{\frac{1}{16}\underset{i,j}{\max}\left\{ \varpi\bigl(e_{i}\widetilde{e}_{j}\bigr),\bigl(\varpi\bigl(e_{i}\widetilde{e}_{j}\bigr)\bigr)^{4}\right\} \times\max_{i,j,p,q}\left\{ \gamma_{ij}\widetilde{\gamma}_{pq}\right\} }\times\varpi\bigl(z\bigr)<1$
then we have $\underset{t\rightarrow+\infty}{\lim}\varpi\circ W^{t}\left(z\right)=0$. 
\end{cor}

\begin{proof}
\textit{a}) If $z\in\mathbb{K}_{+}^{n}\times\mathbb{K}_{+}^{m}$
we have $\varpi\left(z\right)\geq0$, from $\varpi\left(z\right)\leq\tfrac{4}{\underset{i,j}{\max}\left\{ \varpi\left(e_{i}\widetilde{e}_{j}\right)\right\} }$
we get $\varpi\left(z\right)^{2}\leq\tfrac{4}{\underset{i,j}{\max}\left\{ \varpi\left(e_{i}\widetilde{e}_{j}\right)\right\} }\varpi\left(z\right)$
and according to (\ref{eq:woW(z)-maj}) with this we get $\varpi\circ W\left(z\right)\leq\varpi\left(z\right)$
and by induction $\varpi\circ W^{t+1}\left(z\right)\leq\varpi\circ W^{t}\left(z\right)$.\medskip{}

\textit{b}) In proposition \ref{prop:min-woW}, the left-hand
term of the lower bound of $\varpi\circ W^{t}\left(z\right)$
can be put in the form: 
\[
\underset{\left(i,j\right)\in\mathscr{N}}{\min}\left\{ \varpi\left(e_{i}\widetilde{e}_{j}\right)\right\} \left(\underset{\left(i,j\right)\in\mathscr{N}}{\min}\left\{ \sqrt{\gamma_{ij}\widetilde{\gamma}_{ij}}\right\} \right)^{-2}\Biggl(\left(\underset{\left(i,j\right)\in\mathscr{N}}{\min}\left\{ \sqrt{\gamma_{ij}\widetilde{\gamma}_{ij}}\right\} \right)^{2}\biggl(\sum_{i=1}^{n}x_{i}\sum_{j=1}^{m}y_{j}\biggr)\Biggr)^{2^{t-1}}.
\]

\textit{c}) For all $t\geq1$, we have $0\leq\Bigl(\tfrac{1}{4}\varpi\bigl(z\bigr)\Bigr){}^{4^{\left\lfloor \nicefrac{t}{2}\right\rfloor }}\leq\varpi\bigl(z\bigr){}^{4^{\left\lfloor \nicefrac{t}{2}\right\rfloor }}$,
so in proposition \ref{prop:max-woW}, the right-hand term of
the upper bound of $\varpi\circ W^{t}\left(z\right)$ is bounded
by 
\begin{align*}
\left(\underset{i,j}{\max}\left\{ \varpi\left(e_{i}\widetilde{e}_{j}\right)\right\} \right)^{\frac{1}{3}\left(4^{\left\lfloor \nicefrac{\left(t+1\right)}{2}\right\rfloor }-1\right)}\left(\frac{1}{16}\max_{i,j,p,q}\left\{ \gamma_{ij}\widetilde{\gamma}_{pq}\right\} \right)^{\frac{1}{3}\left(4^{\left\lfloor \nicefrac{t}{2}\right\rfloor }-1\right)}\varpi\bigl(z\bigr){}^{4^{\left\lfloor \nicefrac{t}{2}\right\rfloor }}
\end{align*}

that can be written 
\begin{align*}
\left(\frac{1}{16}\underset{i,j}{\max}\left\{ \varpi\left(e_{i}\widetilde{e}_{j}\right)\right\} \max_{i,j,p,q}\left\{ \gamma_{ij}\widetilde{\gamma}_{pq}\right\} \right)^{-\frac{1}{3}}\times\hspace*{6cm}\\
\Biggl(\left(\underset{i,j}{\max}\left\{ \varpi\left(e_{i}\widetilde{e}_{j}\right)\right\} \right)^{\frac{1}{3}\times4^{\left(\left\lfloor \nicefrac{\left(t+1\right)}{2}\right\rfloor -\left\lfloor \nicefrac{t}{2}\right\rfloor \right)}}\left(\frac{1}{16}\max_{i,j,p,q}\left\{ \gamma_{ij}\widetilde{\gamma}_{pq}\right\} \right)^{\frac{1}{3}}\varpi\bigl(z\bigr)\Biggr)^{4^{\left\lfloor \nicefrac{t}{2}\right\rfloor }}
\end{align*}
but $4^{\left(\left\lfloor \nicefrac{\left(t+1\right)}{2}\right\rfloor -\left\lfloor \nicefrac{t}{2}\right\rfloor \right)}=1$
or $4$, so we have 
\[
\left(\underset{i,j}{\max}\left\{ \varpi\left(e_{i}\widetilde{e}_{j}\right)\right\} \right)^{4^{\left(\left\lfloor \nicefrac{\left(t+1\right)}{2}\right\rfloor -\left\lfloor \nicefrac{t}{2}\right\rfloor \right)}}\leq\max\left\{ \underset{i,j}{\max}\left\{ \varpi\left(e_{i}\widetilde{e}_{j}\right)\right\} ,\Bigl(\underset{i,j}{\max}\left\{ \varpi\left(e_{i}\widetilde{e}_{j}\right)\right\} \Bigr)^{4}\right\} 
\]
we also have 
\[
\max\left\{ \underset{i,j}{\max}\left\{ \varpi\left(e_{i}\widetilde{e}_{j}\right)\right\} ,\underset{i,j}{\max}\left\{ \varpi\left(e_{i}\widetilde{e}_{j}\right)\right\} ^{4}\right\} =\underset{i,j}{\max}\left\{ \varpi\left(e_{i}\widetilde{e}_{j}\right),\bigl(\varpi\bigl(e_{i}\widetilde{e}_{j}\bigr)\bigr)^{4}\right\} 
\]
which gives the result. 
\end{proof}
\medskip{}

\section{Normalized gonosomic evolution operators.}

From a gonosomic evolution operator $W$ we define an operator
$V$ which gives the relative frequency distribution of genetic
types.\medskip{}

In this section, we still assume that $\mathbb{K}$ is a commutative
ordered field. For applications in genetics we restrict to the
following simplex of $\mathbb{K}^{n}\times\mathbb{K}^{m}$: 
\[
S^{\:n+m-1}=\left\{ \left(\left(x_{1},\ldots,x_{n}\right),\left(y_{1},\ldots,y_{m}\right)\right)\in\mathbb{K}^{n}\times\mathbb{K}^{m}:x_{i}\geq0,y_{i}\geq0,\sum_{i=1}^{n}x_{i}+\sum_{i=1}^{m}y_{i}=1\right\} 
\]
this simplex is associated with frequency distributions of the
genetic types $e_{i}$ and $\widetilde{e}_{j}$. But as can be
seen from the following example, in general the gonosomic operator
$W$ does not preserve the simplex $S^{\:n+m-1}$.
\begin{example}
In the example \ref{exa:Hybrid-dysgenesis}, with$z=\frac{1}{2}e_{1}+\frac{1}{2}\widetilde{e}_{1}$
and $\mu=0.001$ we get $W\left(\frac{1}{2},\left(\frac{1}{2},0\right)\right)=\left(\frac{1}{8},\left(\frac{999}{8000},\frac{1}{8000}\right)\right)$
and $\varpi\circ W\left(\frac{1}{2},\left(\frac{1}{2},0\right)\right)=\frac{1}{4}$.
\end{example}

For this reason we associate an another operator to $W$. To
define this operator, it is necessary to exclude from the simplex
$S^{\:n+m-1}$ the elements $z\in\mathbb{K}^{n}\times\mathbb{K}^{m}$
such that $\varpi\circ W\left(z\right)=0$. \medskip{}

Using the sets $\mathscr{M}$ and $\mathscr{N}$ introduced in
(\ref{eq:MetN}) we get the following result. 
\begin{prop}
\label{prop:woW(z)=00003D0}Given a non negative gonosomic $\mathbb{K}$-algebra
with $\left(e_{i}\right)_{1\text{\ensuremath{\le}}i\text{\ensuremath{\le}}n}\cup\left(\widetilde{e}_{j}\right)_{1\text{\ensuremath{\le}}j\text{\ensuremath{\le}}m}$
as gonosomic basis. For $z\in\mathbb{K}_{+}^{n}\times\mathbb{K}_{+}^{m}$,
$z=\left(\left(x_{1},\ldots,x_{n}\right),\left(y_{1},\ldots,y_{m}\right)\right)$
we have $\varpi\circ W\left(z\right)=0$ if and only if each
$\left(i,p\right)\in N_{n}\times N_{m}$ verify one of the following
conditions: $\left(i,p\right)\in\mathscr{M}$ or $\left(i,p\right)\in\mathscr{N}$
and $x_{i}y_{p}=0$,.
\end{prop}

\begin{proof}
Using (\ref{eq:masse de x^(t)}) with $t=1$ we get $\varpi\circ W\left(z\right)=\sum_{i,p=1}^{n,m}\varpi\left(e_{i}\widetilde{e}_{p}\right)x_{i}y_{p}$,
we have $\varpi\left(e_{i}\widetilde{e}_{p}\right)x_{i}y_{p}\geq0$
for all $i\in N_{n}$ and $p\in N_{m}$, therefore we get $\varpi\left(e_{i}\widetilde{e}_{p}\right)x_{i}y_{p}=0$
for any $\left(i,p\right)\in N_{n}\times N_{m}$ from which the
result follows.
\end{proof}
This result leads to the definition of the following set
\begin{align*}
\mathcal{O}^{\:n,m} & =\left\{ \left(\left(x_{1},\ldots,x_{n}\right),\left(y_{1},\ldots,y_{m}\right)\right)\in\mathbb{K}_{+}^{n}\times\mathbb{K}_{+}^{m}:x_{i}y_{p}=0,\left(i,p\right)\in\mathscr{N}\right\} .
\end{align*}

For $z=\left(\left(x_{1},\ldots,x_{n}\right),\left(y_{1},\ldots,y_{m}\right)\right)$
such that $x_{i}=0$ for all $i\in N_{n}$ or $y_{p}=0$ for
all $p\in N_{m}$ we have $z\in\mathcal{O}^{\:n,m}$, therefore
the set $\mathcal{O}^{\:n,m}$ is not empty.
\begin{rem}
If $\mathscr{M}=\textrm{Ø}$, we get
\[
\mathcal{O}^{\:n,m}=\left\{ \left(\left(x_{i}\right)_{i\in N_{n}},\left(y_{p}\right)_{p\in N_{m}}\right),x_{1}=\ldots=x_{m}=0\text{ or }y_{1}=\ldots=y_{m}=0\right\} .
\]
Indeed, in this case we have $\mathscr{N}=N_{n}\times N_{m}$,
thus if there is $p\in N_{m}$ such that $y_{p}\neq0$ from $x_{i}y_{p}=0$
we get $x_{i}=0$ for all $i\in N_{n}$. And we get $y_{p}=0$
for all $p\in N_{m}$ if there is $i\in N_{n}$ such that $x_{i}\neq0$.
\end{rem}

\medskip{}

In what follows, we examine the links between the set $\mathcal{O}^{\:n,m}$,
the nilpotency of $W$ and the annihilation of $\varpi\circ W^{t}$.
\begin{prop}
\label{prop:W^(t)}In a non negative gonosomic algebra of type
$\left(n,m\right)$:

a) If there is $t_{0}\geq1$ such that $W^{t_{0}}\bigl(z\bigr)=0$
then $W^{t}\bigl(z\bigr)=0$ for all $t\geq t_{0}$.

b) If there is $t_{0}\geq0$ such that $W^{t_{0}}\left(z\right)\in\mathcal{O}^{\:n,m}$
then $W^{t_{0}+1}\left(z\right)=0$.

c) For $z\in\mathbb{K}_{+}^{n}\times\mathbb{K}_{+}^{m}$ and
$t_{0}\geq0$ we have $W^{t_{0}}\left(z\right)\in\mathcal{O}^{\:n,m}\;\Leftrightarrow\;\varpi\circ W^{t_{0}+1}\left(z\right)=0$.

d) For $z\in\mathbb{K}_{+}^{n}\times\mathbb{K}_{+}^{m}$, $z\neq0$,
if $W^{t}\left(z\right)=0$ then there is $0\leq t_{0}<t$ such
that $W^{t_{0}}\left(z\right)\neq0$ and $W^{t_{0}}\left(z\right)\in\mathcal{O}^{\:n,m}$. 
\end{prop}

\begin{proof}
\textit{a}) Let $W^{t_{0}}\bigl(z\bigr)=\left(\left(x_{1},\ldots,x_{n}\right),\left(y_{1},\ldots,y_{m}\right)\right)$,
from $W^{t_{0}}\bigl(z\bigr)=0$ we deduce that $x_{i}=0$ and
$y_{p}=0$ for all $i$ and $p$ what implies according to (\ref{eq:Op-W})
: $x'_{i}=0$ and $y'_{p}=0$, thus $W^{t_{0}+1}\bigl(z\bigr)=0$
and the result follows by induction.\medskip{}

\textit{b}) With $W^{t_{0}}\left(z\right)=\left(\left(x_{1},\ldots,x_{n}\right),\left(y_{1},\ldots,y_{m}\right)\right)$,
if $W^{t_{0}}\left(z\right)\in\mathcal{O}^{\:n,m}$ we get $x_{i}y_{p}=0$
for all $\left(i,p\right)\in\mathscr{N}$ thus $W^{t_{0}+1}\left(z\right)=\sum_{\left(i,p\right)\in\mathscr{M}}x_{i}y_{p}e_{i}\widetilde{e}_{p}$
but for any $\left(i,p\right)\in\mathscr{M}$ we have $\varpi\left(e_{i}\widetilde{e}_{p}\right)=0$
what implies $e_{i}\widetilde{e}_{p}=0$ and therefore $W^{t_{0}+1}\left(z\right)=0$.\medskip{}

\textit{c}) The necessary condition follows immediately from
\textit{b}). For the sufficiency, let $A$ be a non negative
gonosomic algebra of type $\left(n,m\right)$ and $\left(e_{i}\right)\cup\left(\widetilde{e}_{p}\right)$
a gonosomic basis of $A$. From $z\in\mathbb{K}_{+}^{n}\times\mathbb{K}_{+}^{m}$
and the proposition \ref{prop:W(R+,R+)} we recursively deduce
that $W^{t_{0}+1}\left(z\right)\in\mathbb{K}_{+}^{n}\times\mathbb{K}_{+}^{m}$.
If $W^{t_{0}}\left(z\right)=\left(\left(x_{1},\ldots,x_{n}\right),\left(y_{1},\ldots,y_{m}\right)\right)$
we have with (\ref{eq:masse de x^(t)}) and (\ref{eq:MetN})
$W^{t_{0}+1}\left(z\right)=\sum_{\left(i,p\right)\in\mathscr{M}}\varpi\left(e_{i}\widetilde{e}_{p}\right)x_{i}y_{p}+\sum_{\left(i,p\right)\in\mathscr{N}}\varpi\left(e_{i}\widetilde{e}_{p}\right)x_{i}y_{p}=\sum_{\left(i,p\right)\in\mathscr{N}}\varpi\left(e_{i}\widetilde{e}_{p}\right)x_{i}y_{p}$
therefore if $\varpi\circ W^{t_{0}+1}\left(z\right)=0$ since
$\varpi\left(e_{i}\widetilde{e}_{p}\right)\neq0$ for all $\left(i,p\right)\in\mathscr{N}$,
we get $x_{i}y_{p}=0$ for all $\left(i,p\right)\in\mathscr{N}$,
thus according to the proposition (\ref{prop:woW(z)=00003D0}),
we get $W^{t_{0}}\left(z\right)\in\mathcal{O}^{\:n,m}$.\medskip{}

\textit{d}) Let $z\in\mathbb{K}_{+}^{n}\times\mathbb{K}_{+}^{m}$,
$z\neq0$ such that $W^{t}\left(z\right)=0$, then $t>0$. Let
$t_{0}\geq0$ the smallest integer such that $W^{t_{0}+1}\left(z\right)=0$,
thus $t_{0}+1\leq t$ and $W^{t_{0}}\left(z\right)\neq0$, moreover
according to \textit{c}) we get $W^{t_{0}}\left(z\right)\in\mathcal{O}^{\:n,m}$. 
\end{proof}
\begin{rem}
Genetically, in a bisexual population, concerning a sex-linked
gene, the nilpotency of the operator $W$ means that all genetic
types disappear. According to the result \emph{a}) if all sex-linked
genes disappear from the population they do not reappear. Results
\emph{b}) and \emph{c}) means that for each genetically non-sterile
cross, if the frequency of one of the sex-linked types is zero,
then all types disappear from the population in the next generation.
Finally, result \textit{d}) means that if in a given generation
all the sex-linked types have disappeared, it is because in a
previous generation, for each genetically non-sterile cross,
one of the types had disappeared. 
\end{rem}

Given an gonosomic basis $\left(e_{i}\right)_{1\text{\ensuremath{\le}}i\text{\ensuremath{\le}}n}\cup\left(\widetilde{e}_{p}\right)_{1\text{\ensuremath{\le}}p\text{\ensuremath{\le}}m}$
such that $\gamma_{ipn}=0$ for any $\left(i,p\right)\in N_{n}\times N_{m}$,
according to (\ref{eq:Op-W}) we get $x'_{n}=0$, $x'_{k}=\sum_{i,p=1}^{n-1,m}\gamma_{ipk}x_{i}y_{p}$
and $y'_{r}=\sum_{i,p=1}^{n-1,m}\widetilde{\gamma}_{ipr}x_{i}y_{p}$
for all $k\in N_{n-1}$ and $r\in N_{m}$. From this we conclude
that from the second generation the female type $e_{n}$ disappears
definitively of the population. Furthermore in this case the
evolution operator $W^{2}$ is associated with the gonosomic
algebra of type $\left(n-1,m\right)$ with the gonosomic basis
$\left(e_{i}\right)_{1\text{\ensuremath{\le}}i\text{\ensuremath{\le}}n-1}\cup\left(\widetilde{e}_{p}\right)_{1\text{\ensuremath{\le}}p\text{\ensuremath{\le}}m}$.
We have an analogous conclusion concerning the male type $\widetilde{e}_{m}$
if we assume that $\widetilde{\gamma}_{jqm}=0$ for any $\left(j,q\right)\in N_{n}\times N_{m}$.
\medskip{}

This leads us to give the following definition: 
\begin{defn}
A gonosomic basis $\left(e_{i}\right)_{1\text{\ensuremath{\le}}i\text{\ensuremath{\le}}n}\cup\left(\widetilde{e}_{p}\right)_{1\text{\ensuremath{\le}}p\text{\ensuremath{\le}}m}$
is said to be \textit{irreducible} if it verifies at least one
of the following two conditions:
\begin{align*}
i)\quad\forall k\in N_{n},\exists\left(i,p\right)\in N_{n}\times N_{m};\gamma_{ipk}\neq0,\\
ii)\quad\forall r\in N_{m},\exists\left(j,q\right)\in N_{n}\times N_{m};\widetilde{\gamma}_{jqr}\neq0.
\end{align*}

Otherwise it is said to be \textit{reducible}. And it is said
that a gonosomic algebra is irreducible (resp. reducible) if
its gonosomic base is irreducible (resp. reducible). 
\end{defn}

\begin{example}
The gonosomic algebras given in the examples \ref{exa:male_infertility-1},
\ref{exa:bidirectional}, \ref{exa:Hybrid-dysgenesis} and \ref{exa:partial infertility}
are irreducible. 
\end{example}

From a genetic point of view, as the following proposition shows,
when a gonosomic basis is reducible then at least one of the
female or male genetic types disappears from the population in
the first generation.
\begin{prop}
Let $A$ be a gonosomic algebra of type $\left(n,m\right)$ then
the derived subalgebra $A^{2}$ is gonosomic irreducible of type
$\left(n',m'\right)$ where $n'\leq n$ and $m'\leq m$. 
\end{prop}

\begin{proof}
It is immediate that if $A$ is irreducible then $A^{2}$ is
also gonosomic irreducible. If $A$ is reducible, let the sets
\begin{align*}
L & =\left\{ k\in N_{n};\gamma_{ipk}=0,\forall\left(i,p\right)\in N_{n}\times N_{m}\right\} \\
M & =\left\{ r\in N_{m};\widetilde{\gamma}_{ipr}=0,\forall\left(i,p\right)\in N_{n}\times N_{m}\right\} ,
\end{align*}
according to (\ref{eq:Op-W}) we have $x'_{k}=y'_{r}=0$ for
all $k\in L$ and $r\in L$, it follows that $A^{2}=\text{span}\left\{ \left(e_{i}\right)_{i\in N_{n}\setminus L}\cup\left(\widetilde{e}_{p}\right)_{p\in N_{m}\setminus M}\right\} $
because for all $i\in N_{n}\setminus L$ and $p\in N_{m}\setminus M$
we have $e_{i}\widetilde{e}_{p}=\sum_{k\in N_{n}\setminus L}\gamma_{ipk}e_{k}+\sum_{r\in N_{m}\setminus M}\widetilde{\gamma}_{ipr}\widetilde{e}_{r}$. 
\end{proof}
\begin{prop}
\label{prop:For_def_of_V}Let $A$ be an irreducible non negative
gonosomic $\mathbb{K}$-algebra of type $\left(n,m\right)$ and
$z\in A$, if $W\left(z\right)\in\mathcal{O}^{\:n,m}$ then $z\in\mathcal{O}^{\:n,m}$
. 
\end{prop}

\begin{proof}
Let be $z=\Bigl(\left(x_{i}\right)_{i\in N_{n}}\left(y_{p}\right)_{p\in N_{m}}\Bigr)$
and $W\left(z\right)=\Bigl(\left(x'_{i}\right)_{i\in N_{n}}\left(y'_{p}\right)_{p\in N_{m}}\Bigr)$.
If $W\left(z\right)\in\mathcal{O}^{\:n,m}$ then for all $\left(k,r\right)\in\mathscr{N}$
we have $x'_{k}y'_{r}=0$ and thus $x'_{k}=0$ or $y'_{r}=0$.
According to (\ref{eq:Op-W}), $x'_{k}=\sum_{i,p=1}^{n,m}\gamma_{ipk}x_{i}y_{p}$
and $y'_{r}=\sum_{i,p=1}^{n,m}\widetilde{\gamma}_{ipr}x_{i}y_{p}$.
If $x'_{k}=0$ we have $\gamma_{ipk}x_{i}y_{p}=0$ for all $\left(i,p\right)\in N_{n}\times N_{m}$,
but as $A$ is irreducible there is $\left(j,q\right)\in N_{n}\times N_{m}$
such that $\gamma_{jqk}\neq0$ and thus $\left(j,q\right)\in\mathscr{N}$.
By a similar reasoning we show that $\mathscr{N}\neq\textrm{Ø}$
and therefore that $z\in\mathcal{O}^{\:n,m}$ when we assume
$y'_{r}=0$. 
\end{proof}
\medskip{}

In the following, for any irreducible non negative gonosomic
$\mathbb{K}$-algebra of type $\left(n,m\right)$ we define the
set 
\[
S^{\:n,m}=S^{\:n+m-1}\setminus\mathcal{O}^{\:n,m}
\]
and the operator $V$ called the \textit{normalized gonosomic
operator} of $W$ 
\[
V:S^{\:n,m}\rightarrow S^{\:n,m},\quad z\mapsto\frac{1}{\varpi\circ W\left(z\right)}W\left(z\right).
\]

Using the relations (\ref{eq:Op-W}) we can express the operator
$V$ in coordinate form: 
\begin{equation}
V:\left\{ \begin{aligned}x_{k}' & =\dfrac{\sum_{i,p=1}^{n,m}\gamma_{ipk}x_{i}y_{p}}{\sum_{i,p=1}^{n,m}\varpi\left(e_{i}\widetilde{e}_{p}\right)x_{i}y_{p}},\quad k=1,\ldots,n\medskip\\
y'_{r} & =\dfrac{\sum_{i,p=1}^{n,m}\widetilde{\gamma}_{ipr}x_{i}y_{p}}{\sum_{i,p=1}^{n,m}\varpi\left(e_{i}\widetilde{e}_{p}\right)x_{i}y_{p}},\quad r=1,\ldots,m.
\end{aligned}
\right.\label{eq:Op-V}
\end{equation}
The coordinates of the operator $V$ correspond to the relative
frequencies of genetic types. 
\begin{prop}
\label{prop:V_well_def}The operator $V$ is well defined. 
\end{prop}

\begin{proof}
Indeed, if in the result \textit{c}) of the proposition \ref{prop:W^(t)}
we take $t_{0}=0$, for $z\in A$ we get that $z\notin\mathcal{O}^{\:n,m}\;\Leftrightarrow\;\varpi\circ W\left(z\right)\neq0$. 
\end{proof}
As the following examples show, on the long term, operator $V$
gives results on the evolution of the relative proportions of
the different types which complement the results on the evolution
in population size obtained with operator $W$, particularly
when, under certain conditions, the number of individuals tends
to infinity.
\begin{example}
From the example \ref{exa:W-male_infertility} we get 
\[
V^{t}\left(x_{1},\left(y_{1},y_{2}\right)\right)=\left(\frac{1}{2},\left(\frac{1-\mu}{2},\frac{\mu}{2}\right)\right)
\]
 for all $t\geq1$ and $\left(x_{1},\left(y_{1},y_{2}\right)\right)\in\mathbb{K}\times\mathbb{K}^{2}$
such that $x_{1}y_{1}\neq0$.\medskip{}

Based on the results obtained in the example \ref{exa:W-bidirectional}
we get
\[
\underset{t\rightharpoondown+\infty}{\lim}V^{t}\left(\left(x_{1},x_{2}\right),\left(y_{1},y_{2}\right)\right)=\begin{cases}
\left(\left(\tfrac{1}{2},0\right),\left(\tfrac{1}{2},0\right)\right) & \text{if }0<x_{2}y_{2}<x_{1}y_{1},\\
\left(\left(\tfrac{1}{4},\tfrac{1}{4}\right),\left(\tfrac{1}{4},\tfrac{1}{4}\right)\right) & \text{if }x_{1}y_{1}=x_{2}y_{2}\neq0,\\
\left(\left(0,\tfrac{1}{2}\right),\left(0,\tfrac{1}{2}\right)\right) & \text{if }0<x_{1}y_{1}<x_{2}y_{2}.
\end{cases}
\]

Using the results from the example \ref{exa:W-Hybrid-dysgenesis},
let be
\[
F_{1}\left(t\right)=\frac{x_{1}^{\left(t\right)}}{2x_{1}^{\left(t\right)}+2x_{2}^{\left(t\right)}}=\frac{1}{2+2\frac{x_{2}^{\left(t\right)}}{x_{1}^{\left(t\right)}}}\text{ and }F_{2}\left(t\right)=\frac{x_{2}^{\left(t\right)}}{2x_{1}^{\left(t\right)}+2x_{2}^{\left(t\right)}}=\frac{1}{2\frac{x_{1}^{\left(t\right)}}{x_{2}^{\left(t\right)}}+2}.
\]
With this, for any $t\geq1$ we get $V^{t}\left(\left(x_{1},x_{2}\right),\left(y_{1},y_{2}\right)\right)=\left(\left(F_{1}\left(t\right),F_{2}\left(t\right)\right),\left(F_{1}\left(t\right),F_{2}\left(t\right)\right)\right)$.
When $x_{2}=0$ we have $x_{2}^{\left(t\right)}=y_{2}^{\left(t\right)}$
for all $t\geq1$ and if $x_{1}y_{1}\neq0$ we get $F_{1}\left(t\right)=\frac{1}{2}$
and $F_{2}\left(t\right)=0$. Next, it is clear that if $\underset{t\rightharpoondown+\infty}{\lim}x_{1}^{\left(t\right)}\in\left\{ 0,2\right\} $
and $\underset{t\rightharpoondown+\infty}{\lim}x_{2}^{\left(t\right)}=+\infty$
we get $\underset{t\rightharpoondown+\infty}{\lim}F_{1}\left(t\right)=0$
and $\underset{t\rightharpoondown+\infty}{\lim}F_{2}\left(t\right)=\frac{1}{2}$.
When $\underset{t\rightharpoondown+\infty}{\lim}x_{1}^{\left(t\right)}=\underset{t\rightharpoondown+\infty}{\lim}x_{2}^{\left(t\right)}=+\infty$,
noting that for $t\geq1$ we have $\tfrac{x_{2}^{\left(t+1\right)}}{x_{1}^{\left(t+1\right)}}=\tfrac{\left(x_{1}^{\left(t\right)}+x_{2}^{\left(t\right)}\right)x_{2}^{\left(t\right)}}{\left(x_{1}^{\left(t\right)}\right)^{2}}=\left(\tfrac{x_{2}^{\left(t\right)}}{x_{1}^{\left(t\right)}}\right)^{2}+\tfrac{x_{2}^{\left(t\right)}}{x_{1}^{\left(t\right)}}$,
setting $z^{\left(t\right)}=\frac{x_{2}^{\left(t\right)}}{x_{1}^{\left(t\right)}}$
we get $z^{\left(t+1\right)}=\left(z^{\left(t\right)}\right)^{2}+z^{\left(t\right)}$,
it follows that $\underset{t\rightharpoondown+\infty}{\lim}z^{\left(t\right)}=+\infty$
thus $\underset{t\rightharpoondown+\infty}{\lim}\frac{x_{2}^{\left(t\right)}}{x_{1}^{\left(t\right)}}=+\infty$
and therefore $\underset{t\rightharpoondown+\infty}{\lim}F_{1}\left(t\right)=0$
and $\underset{t\rightharpoondown+\infty}{\lim}F_{2}\left(t\right)=\frac{1}{2}$.
Finally we got\ref{prop:W^(t)}
\[
\underset{t\rightharpoondown+\infty}{\lim}V^{t}\left(\left(x_{1},x_{2}\right),\left(y_{1},y_{2}\right)\right)=\begin{cases}
\left(\left(\tfrac{1}{2},0\right),\left(\tfrac{1}{2},0\right)\right) & \text{if }x_{1}y_{1}\neq0\text{ and }x_{2}=0,\\
\left(\left(0,\tfrac{1}{2}\right),\left(0,\tfrac{1}{2}\right)\right) & \text{if }x_{1}y_{1}<4\text{ and }x_{1}y_{1}+4x_{2}>8\text{ or }4\leq x_{1}y_{1}.
\end{cases}
\]

It is observed that if the population contains females of type
$P$, then in the event of a rise in the temperature of the environment,
type $M$ disappears from the population and is replaced by type
$P$, which ensures the survival of the species.

\medskip{}
\end{example}

\begin{lem}
\label{lem:woW<>0}In an irreducible non negative gonosomic algebra,
for all $z\in S^{\:n,m}$ and $t\geq1$ we have $\varpi\circ W^{t}\left(z\right)\neq0$.
\end{lem}

\begin{proof}
By induction on $t\geq1$. According to the proposition (\ref{prop:V_well_def})
we have $\varpi\circ W\left(z\right)\neq0$. Suppose that $\varpi\circ W^{k}\left(z\right)\neq0$
for all $1\leq k\leq t$, then we have $\varpi\circ W^{t+1}\left(z\right)\neq0$,
indeed if we suppose that $\varpi\circ W^{t+1}\left(z\right)=0$,
according to the result \emph{c}) of the proposition (\ref{prop:W^(t)})
we get $W^{t}\left(z\right)\in\mathcal{O}^{\:n,m}$, but from
the proposition (\ref{prop:For_def_of_V}) it follows that $W^{t-1}\left(z\right)\in\mathcal{O}^{\:n,m}$,
which implies, according to the result \emph{c}) of the proposition
(\ref{prop:W^(t)}) that $\varpi\circ W^{t}\left(z\right)=0$.,
contradiction
\end{proof}
There is a relation between the dynamics of the operators $V$
and $W$. 
\begin{prop}
\label{prop:V^t}In an irreducible non negative gonosomic algebra,
for all $z\in S^{\:n,m}$ and $t\geq0$ we have 
\begin{align*}
a)\quad & V^{t}\left(z\right)=\frac{1}{\varpi\circ W^{t}\left(z\right)}W^{t}\left(z\right),\\
b)\quad & V^{t}\left(\lambda z\right)=V^{t}\left(z\right),\qquad\left(\forall\lambda\in\mathbb{K},\lambda\neq0\right),\\
c)\quad & V^{t}\left(z\right)\neq0.
\end{align*}
\end{prop}

\begin{proof}
\textit{a}) By induction on $t\geq1$. Suppose that $V^{t}\left(z\right)=\frac{1}{\varpi\circ W^{t}\left(z\right)}W^{t}\left(z\right)$
for a $t\geq1$, according to the lemma \ref{lem:woW<>0} we
have $\varpi\circ W^{t+1}\left(z\right)\neq0$. Next we have
\[
W\left(V^{t}\left(z\right)\right)=\frac{1}{\left(\varpi\circ W^{t}\left(z\right)\right)^{2}}W^{t+1}\left(z\right),\quad\left(*\right)
\]
from this we get
\[
\varpi\circ W\left(V^{t}\left(z\right)\right)=\frac{1}{\left(\varpi\circ W^{t}\left(z\right)\right)^{2}}\varpi\circ W^{t+1}\left(z\right)\quad\left(**\right)
\]
 thus $\varpi\circ W\left(V^{t}\left(z\right)\right)\neq0$.
By the definition of the operator $V$ we have $V^{t+1}\left(z\right)=V\left(V^{t}\left(z\right)\right)=\frac{1}{\varpi\circ W\left(V^{t}\left(z\right)\right)}W\left(V^{t}\left(z\right)\right)$
and using the results $\left(*\right)$ and $\left(**\right)$
we get the relation to the order $t+1$.

\textit{b}) For all $\lambda\in\mathbb{K},\lambda\neq0$ and
$t\geq0$ we have $W^{t}\left(\lambda z\right)=\left(\frac{1}{2}\right)^{2^{t}-1}\lambda^{2^{t}}W^{t}\left(z\right)$
thus $\varpi\circ W^{t}\left(\lambda z\right)=\left(\frac{1}{2}\right)^{2^{t}-1}\lambda^{2^{t}}\varpi\circ W^{t}\left(z\right)$,
therefore if $\varpi\circ W^{t}\left(z\right)\neq0$ we have
also $\varpi\circ W^{t}\left(\lambda z\right)\neq0$ and with
the above result we get $V^{t}\left(\lambda z\right)=V^{t}\left(z\right)$.

\textit{c}) This results from the facts that $V^{t}\left(z\right)\notin\mathcal{O}^{\:n,m}$
and $\left(\left(0\right)_{n},\left(0\right)_{m}\right)\in\mathcal{O}^{\:n,m}$. 
\end{proof}
\begin{prop}
Let $A$ be an irreducible non negative gonosomic algebra and
$\left(e_{i}\right)_{1\text{\ensuremath{\le}}i\text{\ensuremath{\le}}n}\cup\left(\widetilde{e}_{j}\right)_{1\text{\ensuremath{\le}}j\text{\ensuremath{\le}}m}$
a gonosomic basis of $A$.

For all $z\in S^{\:n,m}$ and $t\geq1$ we note $V^{t}\bigl(z\bigr)=\left(v_{1}^{\:\left(t\right)},\ldots,v_{n}^{\:\left(t\right)},w_{1}^{\;\left(t\right)},\ldots,w_{m}^{\;\left(t\right)}\right)$.
Let the sets be $E_{k}=\left\{ \left(i,p\right)\in N_{n}\times N_{m};\gamma_{ipk}>0\right\} $
and $\widetilde{E}_{r}=\left\{ \left(i,p\right)\in N_{n}\times N_{m};\widetilde{\gamma}_{ipr}>0\right\} $,
we have
\[
\dfrac{\underset{\left(i,p\right)\in E_{k}}{\min}\left\{ \gamma_{ipk}\right\} }{\max_{i,p}\left\{ \varpi\left(e_{i}\widetilde{e}_{p}\right)\right\} }\leq v_{k}^{\:\left(t\right)}\leq\dfrac{\max_{i,p}\left\{ \gamma_{ipk}\right\} }{\underset{\left(i,p\right)\in\mathscr{N}}{\min}\left\{ \varpi\left(e_{i}\widetilde{e}_{p}\right)\right\} },
\]
and 
\[
\dfrac{\underset{\left(i,p\right)\in\widetilde{E}_{k}}{\min}\left\{ \widetilde{\gamma}_{ipr}\right\} }{\max_{i,p}\left\{ \varpi\left(e_{i}\widetilde{e}_{p}\right)\right\} }\leq w_{r}^{\;\left(t\right)}\leq\dfrac{\max_{i,p}\left\{ \widetilde{\gamma}_{ipr}\right\} }{\underset{\left(i,p\right)\in\mathscr{N}}{\min}\left\{ \varpi\left(e_{i}\widetilde{e}_{p}\right)\right\} }.
\]
\end{prop}

\begin{proof}
For any $t\geq1$ we note $W^{t-1}\left(z\right)=\sum_{i=1}^{n}x_{i}^{\left(t-1\right)}e_{i}+\sum_{p=1}^{m}y_{p}^{\left(t-1\right)}\widetilde{e}_{p}$.
For all $k\in N_{n}$ we have $v_{k}^{\left(t\right)}=\frac{x_{k}^{\left(t\right)}}{\varpi\circ W^{t}\left(z\right)}$
According to \ref{eq:masse de x^(t)} we have $\varpi\circ W^{t}\left(z\right)=\sum_{i,p=1}^{n,m}\varpi\left(e_{i}\widetilde{e}_{p}\right)x_{i}^{\left(t-1\right)}y_{p}^{\left(t-1\right)}$,
it follows that 
\begin{align*}
0<\underset{\left(i,p\right)\in\mathscr{N}}{\min}\left\{ \varpi\left(e_{i}\widetilde{e}_{p}\right)\right\} \Bigl(\sum_{i,p=1}^{n,m}x_{i}^{\left(t-1\right)}y_{p}^{\left(t-1\right)}\Bigr) & \leq\varpi\circ W^{t}\left(z\right)\leq\underset{i,p}{\max}\left\{ \varpi\left(e_{i}\widetilde{e}_{p}\right)\right\} \Bigl(\sum_{i,p=1}^{n,m}x_{i}^{\left(t-1\right)}y_{p}^{\left(t-1\right)}\Bigr),
\end{align*}
and by using the inequalities given in the proposition \ref{prop:x_k and y_k}
we get the inequalities for $v_{r}^{\;\left(t\right)}$ . A similar
reasoning gives the inequalities for $w_{r}^{\;\left(t\right)}$. 
\end{proof}
\medskip{}

There is a relation between the fixed points of the operators
$W$ and $V$, for this we introduce the following definition. 
\begin{defn}
A point $z=\left(\left(x_{1},\ldots,x_{n}\right),\left(y_{1},\ldots,y_{m}\right)\right)\in\mathbb{K}^{n}\times\mathbb{K}^{m}$
of a gonosomic algebra of type $\left(n,m\right)$ is non-negative
and \emph{normalizable} if it satisfies the following conditions
$x_{i},y_{j}\geq0$ and $\sum_{i=1}^{n}x_{i}+\sum_{j=1}^{m}y_{j}\neq0$. 
\end{defn}

A consequence of this definition is that for any non-negative
and normalizable point $z$ we have $\varpi\left(z\right)\neq0$. 
\begin{prop}
In an irreducible non negative gonosomic $\mathbb{K}$-algebra,
the map $z^{*}\mapsto\frac{1}{\varpi\left(z^{*}\right)}z^{*}$
is an one-to-one correspondence between the set of non-negative
and normalizable fixed point of $W$ and the set of fixed points
of the operator $V$. 
\end{prop}

\begin{proof}
Let $A$ be an irreducible non negative gonosomic algebra. If
$z^{*}\in A$ verifies $W\left(z^{*}\right)=z^{*}$ then first
$\varpi\circ W\left(z^{*}\right)=\varpi\left(z^{*}\right)\neq0$,
next $W\left(\frac{1}{\varpi\left(z^{*}\right)}z^{*}\right)=\frac{1}{2\varpi\left(z^{*}\right)^{2}}z^{*}$
thus $\varpi\circ W\left(\frac{1}{\varpi\left(z^{*}\right)}z^{*}\right)=\frac{1}{2\varpi\left(x^{*}\right)}$
therefore $V\left(\frac{1}{\varpi\left(z^{*}\right)}z^{*}\right)=\frac{1}{\varpi\left(z^{*}\right)}z^{*}$
which proves that $\frac{1}{\varpi\left(z^{*}\right)}z^{*}$
is a fixed point of $V$. Therefore the map $\Psi:S^{n,m}\rightarrow S^{n,m}$,
$z\mapsto\frac{z}{\varpi\left(z\right)}$ maps the set of fixed
points for $W$ and $V$. Let $x^{*}$ and $y^{*}$ two fixed
points of $W$ such that $\Psi\left(x^{*}\right)=\Psi\left(y^{*}\right)$
by applying $\varpi\circ W$ to this equality we get $\frac{1}{\varpi\left(x^{*}\right)}=\frac{1}{\varpi\left(y^{*}\right)}$,
with this we deduce from $\Psi\left(x^{*}\right)=\Psi\left(y^{*}\right)$
that $x^{*}=y^{*}$, therefore $\Psi$ is one-to-one. Next, given
$z^{*}\in A$ such that $V\left(z^{*}\right)=z^{*}$ then we
have $\varpi\circ W\left(z^{*}\right)\neq0$ and dividing the
two members of this equality by $\varpi\circ W\left(z^{*}\right)$
we get $\frac{1}{\varpi\circ W\left(z^{*}\right)^{2}}W\left(z^{*}\right)=\frac{1}{\varpi\circ W\left(z^{*}\right)}z^{*}$
in other words $\frac{1}{\varpi\circ W\left(z^{*}\right)}z^{*}$
is a fixed point of $W$, therefore $\Psi$ is surjective.
\end{proof}
\begin{example}
In example \ref{exa:W-male_infertility}, the fixed points of
$W$are $\left(0,\left(0,0\right)\right)$ and $\left(\frac{2}{1-\mu},\left(2,\frac{2\mu}{1-\mu}\right)\right)$,
therefore the fixed point of $V$ is $\left(\frac{1}{2},\left(\frac{1-\mu}{2},\frac{\mu}{2}\right)\right)$.\medskip{}

In example \ref{exa:W-bidirectional}, the fixed points of $W$are
$\left(\left(0,0\right),\left(0,0\right)\right)$, $\left(\left(2,0\right),\left(2,0\right)\right)$,
$\left(\left(0,2\right),\left(0,2\right)\right)$ and $\left(\left(2,2\right),\left(2,2\right)\right)$
it follows that the fixed points of $V$ are $\left(\left(\frac{1}{2},0\right),\left(\frac{1}{2},0\right)\right)$,
$\left(\left(0,\frac{1}{2}\right),\left(0,\frac{1}{2}\right)\right)$
and $\left(\left(\frac{1}{4},\frac{1}{4}\right),\left(\frac{1}{4},\frac{1}{4}\right)\right)$.\medskip{}

For example \ref{exa:W-Hybrid-dysgenesis}, the fixed points
of $W$are $\left(\left(0,0\right),\left(0,0\right)\right)$,
$\left(\left(2,0\right),\left(2,0\right)\right)$ and $\left(\left(0,2\right),\left(0,2\right)\right)$,
therefore the fixed point of $V$ are $\left(\left(\frac{1}{2},0\right),\left(\frac{1}{2},0\right)\right)$
and $\left(\left(0,\frac{1}{2}\right),\left(0,\frac{1}{2}\right)\right)$.

\medskip{}

In example \ref{exa:partial infertility}, we find that $W$
has the following six fixed points: \smallskip{}

\begin{tabular}{ccccc}
$\left(\left(0,0\right),\left(0,0\right)\right)$, &  & $\left(\left(2,0\right),\left(2,0\right)\right)$,\medskip &  & \tabularnewline
$\left(\left(\frac{2}{1-\mu},0\right),\left(0,\frac{2}{1-\mu}\right)\right)$ &  & $\left(\left(0,\frac{2}{1-\nu}\right),\left(\frac{2}{1-\nu},0\right)\right)$ &  & $\left(\left(0,\frac{2}{1-\tau}\right),\left(0,\frac{2}{1-\tau}\right)\right),$\medskip\tabularnewline
\multicolumn{5}{c}{$\left(\left(\frac{2\left(\nu-\tau\right)}{\left(1-\tau\right)-\left(1-\mu\right)\left(1-\nu\right)},\frac{2\mu}{\left(1-\tau\right)-\left(1-\mu\right)\left(1-\nu\right)}\right),\left(\frac{2\left(\mu-\tau\right)}{\left(1-\tau\right)-\left(1-\mu\right)\left(1-\nu\right)},\frac{2\nu}{\left(1-\tau\right)-\left(1-\mu\right)\left(1-\nu\right)}\right)\right).$}\tabularnewline
\end{tabular}\medskip{}

Therefore the fixed points of $V$ are \smallskip{}

\begin{tabular}{ccccccc}
$\left(\left(\frac{1}{2},0\right),\left(\frac{1}{2},0\right)\right)$, &  & $\left(\left(\frac{1}{2},0\right),\left(0,\frac{1}{2}\right)\right)$, &  & $\left(\left(0,\frac{1}{2}\right),\left(\frac{1}{2},0\right)\right)$ &  & $\left(\left(0,\frac{1}{2}\right),\left(0,\frac{1}{2}\right)\right),$\medskip\tabularnewline
\multicolumn{7}{c}{$\left(\left(\frac{\nu-\tau}{2\left(\mu+\nu-\tau\right)},\frac{\mu}{2\left(\mu+\nu-\tau\right)}\right),\left(\frac{\mu-\tau}{2\left(\mu+\nu-\tau\right)},\frac{\nu}{2\left(\mu+\nu-\tau\right)}\right)\right).$}\tabularnewline
\end{tabular}

\medskip{}
\end{example}

The various stability notions of the equilibrium points for $W$
are preserved for $V$. 
\begin{prop}
\label{prop:Notions_Equilibrium_pts}Let $z^{*}$ be a non-negative
and normalizable fixed point of $W$.

a) If $z^{*}$ is periodic with least period $p$ then $\frac{1}{\varpi\left(z^{*}\right)}z^{*}$
is a periodic equilibrium point of the operator $V$ with least
period $p$ .

b) If $z^{*}$ is attracting then $\frac{1}{\varpi\left(z^{*}\right)}z^{*}$
is an attracting equilibrium point of $V$.

c) If $z^{*}$ is stable (resp. uniformly stable) then $\frac{1}{\varpi\left(z^{*}\right)}z^{*}$
is a stable (resp. uniformly stable) equilibrium point of $V$.

d) If $z^{*}$ is asymptotically stable then the fixed point
$\frac{1}{\varpi\left(z^{*}\right)}z^{*}$ of $V$ is asymptotically
stable.

e) If $z^{*}$ is exponentially stable then the fixed point $\frac{1}{\varpi\left(z^{*}\right)}z^{*}$
of $V$ is exponentially stable. 
\end{prop}

\begin{proof}
a) For any integer $t\geq0$, from proposition \ref{prop:W(R+,R+)}
we deduce that $W^{t}\left(z^{*}\right)$ is non negative. If
$z^{*}$ is periodic there is a smaller integer $p$ such that
$W^{p}\left(z^{*}\right)=z^{*}$ it follows that $\varpi\circ W^{p}\left(z^{*}\right)=\varpi\left(z^{*}\right)\neq0$
. Therefore $W^{p}\left(z^{*}\right)$ is non-negative and normalizable.
Using proposition \ref{prop:V^t} we have $V^{p}\left(\frac{1}{\varpi\left(z^{*}\right)}z^{*}\right)=V^{p}\left(z^{*}\right)=\frac{1}{\varpi\circ W^{p}\left(z^{*}\right)}W^{p}\left(z^{*}\right)=\frac{1}{\varpi\left(z^{*}\right)}z^{*}$.
Now let us show that $p$ is the smallest integer verifying this
relation. If it exists $1\leq m<p$ such that $V^{m}\left(\frac{1}{\varpi\left(z^{*}\right)}z^{*}\right)=\frac{1}{\varpi\left(z^{*}\right)}z^{*}$,
we know that $m$ divides $p$, let $p=mq$ with $q\geq2$, according
to proposition \ref{prop:V^t} we have $V^{m}\left(\frac{1}{\varpi\left(z^{*}\right)}z^{*}\right)=V^{m}\left(z^{*}\right)=\frac{W^{m}\left(z^{*}\right)}{\varpi\circ W^{m}\left(z^{*}\right)}$,
thus $W^{m}\left(z^{*}\right)=\frac{\varpi\circ W^{m}\left(z^{*}\right)}{\varpi\left(z^{*}\right)}z^{*}$.
We get 
\begin{align*}
z^{*} & =W^{p}\left(z^{*}\right)=W^{mq}\left(z^{*}\right)=W^{m\left(q-1\right)}\left(\frac{\varpi\circ W^{m}\left(z^{*}\right)}{\varpi\left(z^{*}\right)}z^{*}\right)=\left(\frac{\varpi\circ W^{m}\left(z^{*}\right)}{\varpi\left(z^{*}\right)}\right)^{2^{m\left(q-1\right)}}z^{*}
\end{align*}
but $z^{*}\neq0$ and $\frac{\varpi\circ W^{m}\left(z^{*}\right)}{\varpi\left(z^{*}\right)}\in\mathbb{K}_{+}$
therefore $\frac{\varpi\circ W^{m}\left(z^{*}\right)}{\varpi\left(z^{*}\right)}=1$
and thus we get $W^{m}\left(z^{*}\right)=z^{*}$ with $m<p$,
contradiction.

\medskip{}

From now the space $\mathbb{K}^{n}\times\mathbb{K}^{m}$ is equipped
with the norm $\left\Vert \left(\left(x_{1},\ldots,x_{n}\right),\left(y_{1},\ldots,y_{m}\right)\right)\right\Vert =\sum_{i=1}^{n}\left|x_{i}\right|+\sum_{i=1}^{m}\left|y_{i}\right|$,
for this norm we have $\left\Vert z\right\Vert =\varpi\left(z\right)$
for any $z\in\mathbb{K}_{+}^{n}\times\mathbb{K}_{+}^{m}$.\medskip{}

\textit{b}) If $z^{*}$ is an attractive point of $W$, there
is $\rho>0$ such that for all $z\in\mathbb{K}^{n}\times\mathbb{K}^{m}$
verifying $\left\Vert z-z^{*}\right\Vert <\rho$ we have $\lim_{t\rightarrow\infty}W^{t}\left(z\right)=z^{*}$.
First, as $z^{*}\neq0$ is normalizable we have $\varpi\left(z^{*}\right)\neq0$
and by continuity of $\varpi$ we get $\lim_{t\rightarrow\infty}\varpi\circ W^{t}\left(z\right)=\varpi\left(z^{*}\right)$.
Next, for all $z\in\mathbb{K}^{n}\times\mathbb{K}^{m}$ such
that $\lim_{t\rightarrow\infty}W^{t}\left(z\right)=z^{*}$ we
have $W^{t}\left(z\right)\neq0$ for every $t\geq0$, otherwise
according to the result \emph{a}) of proposition \ref{prop:W^(t)},
it would exist $t_{0}\geq1$ such that $W^{t}\left(z\right)=0$
for $t\geq t_{0}$ and consequently $\lim_{t\rightarrow\infty}W^{t}\left(z\right)=0$,
we deduce that, in particular if $z\in S^{\:n,m}$ we get $\varpi\circ W^{t}\left(z\right)\neq0$.
Finally, for any $z\in S^{\:n,m}$ such that $\left\Vert z-z^{*}\right\Vert <\rho$
we get $\lim_{t\rightarrow\infty}V^{t}\left(z\right)=\lim_{t\rightarrow\infty}\frac{1}{\varpi\circ W^{t}\left(z\right)}W^{t}\left(z\right)=\frac{1}{\varpi\left(z^{*}\right)}z^{*}$.\medskip{}

\emph{c}) By definition, the equilibrium point $z^{*}$ for $W$
is stable if for all $t_{0}\geq0$ and $\epsilon>0$, there exists
$\delta>0$ such that the condition $\left\Vert z-z^{*}\right\Vert <\delta$
implies $\left\Vert W^{t}\left(z\right)-z^{*}\right\Vert <\varepsilon\;\left(t\geq t_{0}\right)$,
and $z^{*}$ is uniformly stable if the existence of $\delta>0$
does not depend on $t_{0}$.

We deduce from proposition \ref{prop:max-woW} that $\varpi\left(z^{*}\right)\geq\frac{4}{\max_{i,j}\left\{ \varpi\left(e_{i}\widetilde{e}_{j}\right)\right\} }$,
in what follows we take $0<\varepsilon<\varpi\left(z^{*}\right)-\frac{2}{\max_{i,j}\left\{ \varpi\left(e_{i}\widetilde{e}_{j}\right)\right\} }$.
For all $z\in S^{\:n,m}$ we get: 
\[
\left\Vert V^{t}\left(z\right)-V\left(z^{*}\right)\right\Vert \leq\left\Vert \tfrac{1}{\varpi\circ W^{t}\left(z\right)}W^{t}\left(z\right)-\tfrac{1}{\varpi\circ W^{t}\left(z\right)}z^{*}\right\Vert +\left\Vert \tfrac{1}{\varpi\circ W^{t}\left(z\right)}z^{*}-\tfrac{1}{\varpi\left(z^{*}\right)}z^{*}\right\Vert 
\]
or 
\begin{equation}
\left\Vert V^{t}\left(z\right)-V\left(z^{*}\right)\right\Vert \leq\tfrac{1}{\varpi\circ W^{t}\left(z\right)}\left\Vert W^{t}\left(z\right)-z^{*}\right\Vert +\left|\tfrac{\varpi\circ W^{t}\left(z\right)-\varpi\left(z^{*}\right)}{\varpi\circ W^{t}\left(z\right)\times\varpi\left(z^{*}\right)}\right|\left\Vert z^{*}\right\Vert .\label{eq:inegal1}
\end{equation}
If we denote $W^{t}\left(z\right)=\bigl(z_{i}^{\:\left(t\right)}\bigr)_{1\leq i\leq n+m}$
and $z^{*}=\left(z_{i}^{*}\right)_{1\leq i\leq n+m}$, we note
that 
\[
\bigl|\varpi\circ W^{t}\left(z\right)-\varpi\left(z^{*}\right)\bigr|\leq\sum_{i=1}^{n+m}\bigl|z_{i}^{\:\left(t\right)}-z_{i}^{*}\bigr|=\left\Vert W^{t}\left(z\right)-z^{*}\right\Vert ,
\]
from this and the hypothesis, we deduce that for all $z\in S^{\:n,m}$
such that $\left\Vert z-z^{*}\right\Vert <\delta$ we have $0<\varpi\left(z^{*}\right)-\varepsilon\leq\varpi\circ W^{t}\left(z\right)$,
with this and $\left\Vert z^{*}\right\Vert =\varpi\left(z^{*}\right)$
the inequality (\ref{eq:inegal1}) becomes 
\[
\left\Vert V^{t}\left(z\right)-V\left(z^{*}\right)\right\Vert \leq\tfrac{2\varepsilon}{\varpi\left(z^{*}\right)-\varepsilon}<\varepsilon\times\max_{i,j}\left\{ \varpi\left(e_{i}\widetilde{e}_{j}\right)\right\} 
\]
the result follows.

\medskip{}

\emph{d}) If $z^{*}$ is asymptotically stable for $W$, by definition
$z^{*}$ is attractive and stable for $W$ but from \emph{b})
and \emph{c}) it follows that $z^{*}$ is attractive and stable
for $V$, thus $z^{*}$ is asymptotically stable for $V$.

\medskip{}

\emph{e}) By definition, the equilibrium point $z^{*}$ of $W$
is exponentially stable if for all $t_{0}\geq0$ there exists
$\delta>0$, $M>0$ and $\eta\in\left]0,1\right[$ such that
for $z\in\mathbb{K}_{+}^{n}\times\mathbb{K}_{+}^{m}$ : 
\[
\left\Vert z-z^{*}\right\Vert \leq\delta\Rightarrow\left\Vert W^{t}\left(z\right)-z^{*}\right\Vert \leq M\eta^{t}\left\Vert z-z^{*}\right\Vert ,\;\mbox{for all }t\geq t_{0}.
\]
Analogously to what was done in \emph{c}), for all $x\in S^{\:n,m}$
we have the inequality: 
\begin{equation}
\left\Vert V^{t}\left(z\right)-V\left(z^{*}\right)\right\Vert \leq\tfrac{1}{\varpi\circ W^{t}\left(z\right)}\left\Vert W^{t}\left(z\right)-z^{*}\right\Vert +\left|\tfrac{\varpi\circ W^{t}\left(z\right)-\varpi\left(z^{*}\right)}{\varpi\circ W^{t}\left(z\right)\times\varpi\left(z^{*}\right)}\right|\left\Vert z^{*}\right\Vert .\label{eq:inegal2}
\end{equation}

As in \emph{c}) we have $\bigl|\varpi\circ W^{t}\left(z\right)-\varpi\left(z^{*}\right)\bigr|\leq\left\Vert W^{t}\left(z\right)-z^{*}\right\Vert $,
we deduce that for all $z\in S^{\:n,m}$ verifying $\left\Vert z-z^{*}\right\Vert \leq\delta$
we get $\varpi\left(z^{*}\right)-M\eta^{t}\left\Vert z-z^{*}\right\Vert \leq\varpi\circ W^{t}\left(z\right)$.
But $\eta\in\left]0,1\right[$, thus there exists $t_{1}\geq t_{0}$
such that $\frac{4}{\max_{i,j}\left\{ \varpi\left(e_{i}\widetilde{e}_{j}\right)\right\} }-M\eta^{t}\left\Vert z-z^{*}\right\Vert \geq\frac{2}{\max_{i,j}\left\{ \varpi\left(e_{i}\widetilde{e}_{j}\right)\right\} }$
for all $t\geq t_{1}$, but according to the result \emph{a})
of proposition \ref{prop:max-woW} we have $\varpi\left(z^{*}\right)\geq\frac{4}{\max_{i,j}\left\{ \varpi\left(e_{i}\widetilde{e}_{j}\right)\right\} }$,
thus for all $z\in S^{\:n,m}$ such that $\left\Vert z-z^{*}\right\Vert \leq\delta$
and for every $t\geq t_{1}$ we have 
\[
0<\frac{2}{\max_{i,j}\left\{ \varpi\left(e_{i}\widetilde{e}_{j}\right)\right\} }\leq\varpi\left(z^{*}\right)-M\eta^{t}\left\Vert z-z^{*}\right\Vert \leq\varpi\circ W^{t}\left(z\right)
\]
with this and $\left\Vert z^{*}\right\Vert =\varpi\left(z^{*}\right)$,
inequality (\ref{eq:inegal2}) becomes 
\[
\left\Vert V^{t}\left(z\right)-V\left(z^{*}\right)\right\Vert \leq\frac{2M\eta^{t}\left\Vert z-z^{*}\right\Vert }{\varpi\left(z^{*}\right)-M\eta^{t}\left\Vert z-z^{*}\right\Vert }\leq M\max_{i,j}\left\{ \varpi\left(e_{i}\widetilde{e}_{j}\right)\right\} \times\eta^{t}\left\Vert z-z^{*}\right\Vert ,\mbox{ for all }t\geq t_{1},
\]
which proves that $z^{*}$ is an exponentially stable point for
$V$. 
\end{proof}
\medskip{}

{\footnotesize\textsc{{}R.\ Varro\smallskip{}
 }}{\footnotesize\par}

{\footnotesize\textsc{{}Institut Montpelliérain Alexander Grothendieck,
Université de Montpellier, }}\\
 {\footnotesize\textsc{35095 Montpellier Cedex 5, France.}}{\footnotesize\par}

{\footnotesize\textsc{richard.varro@umontpellier.fr\smallskip{}
 }}{\footnotesize\par}

{\footnotesize\textsc{Université Montpellier Paul Valéry , Route
de Mende }}\\
 {\footnotesize\textsc{34199 Montpellier cedex 5, France}}{\footnotesize\par}

{\footnotesize\textsc{{}richard.varro@univ-montp3.fr}}{\footnotesize\par}

\end{document}